\documentclass[authoryear,preprint,12pt]{elsarticle}

\usepackage{amssymb}
\usepackage{amsmath}
\usepackage{amsthm}
\newtheorem{proposition}{Proposition}[section]
\newtheorem{corollary}{Corollary}[section]
\usepackage{graphicx}
\PassOptionsToPackage{hyphens}{url}
\usepackage{hyperref}
\usepackage{tabularx}
\usepackage{longtable}
\usepackage{booktabs}
\usepackage[font=footnotesize]{caption}
\usepackage{url}
\usepackage{xcolor}

\graphicspath{{./}}

\newcommand{\manuscripttitle}{Beyond Consistent Scenarios: Deriving Indirect Influence, Transition Resistance, and Adjustment Dynamics}

\begin{document}

\begin{titlepage}
    \begin{center}
        \vspace*{0.1cm}
        {\LARGE Beyond Consistent Scenarios: Deriving Indirect Influence, Transition Resistance, and Adjustment Dynamics}\par
        \vspace{1.0cm}
        \vspace{1.0cm}
                Andrew G. Ross $^{1,*}$, Julia Gershenzon $^{1}$, Andreas Kleefeld $^{2,3,4}$ \\
        \vspace*{1.25\baselineskip}
        $^1$ Forschungszentrum J\"ulich GmbH, Institute of Climate and Energy Systems, J\"ulich Systems Analysis, Wilhelm-Johnen-Strasse, 52425 J\"ulich, Germany \\

\vspace*{0.75\baselineskip}

        $^2$ Forschungszentrum J\"ulich GmbH, J\"{u}lich Supercomputing Centre, Wilhelm-Johnen-Strasse, 52425 J\"ulich, Germany \\

\vspace*{0.75\baselineskip}

        $^3$ FH Aachen - University of Applied Sciences, Faculty of Medical Engineering and Technomathematics, Heinrich-Mu\ss{}mann-Str. 1, 52428 J\"{u}lich, Germany \\

\vspace*{0.75\baselineskip}

        $^4$ FH Aachen - University of Applied Sciences, Institute for Data-driven Technologies, Heinrich-Mu\ss{}mann-Str. 1, 52428 J\"{u}lich, Germany \\

        \vspace*{1\baselineskip}

\vspace*{4\baselineskip}
$^*$ Corresponding author: a.ross@fz-juelich.de \\
\vspace*{4\baselineskip}
July 11, 2026
        
    \end{center}
\end{titlepage}

\begin{titlepage}
    \begin{center}
        \vspace*{0.1cm}
        {\LARGE Beyond Consistent Scenarios: Deriving Indirect Influence, Transition Resistance, and Adjustment Dynamics}\par
        \vspace{0.5cm}
    \end{center}

    \vspace{0.5cm}

    {\noindent Assessments of structural change and economic transition dynamics, such as those arising in the energy transition, depend on internally consistent qualitative scenarios specifying the policy environment, technology mix, governance arrangements, and demand conditions. Cross-Impact Balance (CIB) analysis derives such socio-technical scenarios as fixed-point attractors of an expert-elicited interdependency network, supplying structural inputs upon which assessment models (including energy system optimisation, agent-based, and general equilibrium frameworks) can draw. Standard CIB, however, delivers only this equilibrium catalogue, leaving four structural questions unanswered: how much network-weighted effort a given transition requires; which components are the true system-wide levers once indirect influence chains are counted; in what sequence the system adjusts; and how the network at a given attractor responds to an external shock. This paper extends CIB through Linear Response Theory, exploiting a structural isomorphism between the CIB drift matrix and the Leontief input-output technology matrix. Four analytical objects are derived in closed form: the Type I cross-impact multiplier, which aggregates all direct and indirect influence chains; the perturbation budget, a network-weighted and directionally asymmetric measure of transition effort; the impulse response function, which traces descriptor adjustment sequences and feedback-induced overshoots; and the unit-impulse shock profile, which characterises attractor-specific network sensitivity and yields a direct measure of structural resilience and susceptibility. The framework is applied empirically to an energy-transition cross-impact matrix, yielding all four objects for five structural equilibria, and transfers to any domain in which pairwise influence scores encode structural interdependencies.} \\

    \vspace{0.5em}


    \vspace{0.5em}


    \vfill

\end{titlepage}

\newpage

\section{Introduction}
\label{sec:intro}

\noindent A large structural shift, such as a transition from fossil-fuel dependence to a renewable energy system, is not merely a change in the levels of aggregate variables but a reorganisation of the interdependency structure among the components of the underlying system. The questions at stake concern which variables reinforce which, which configurations are self-sustaining, and how the system responds to external perturbation \citep{MillerPage2009}. The analytical questions that follow are correspondingly structural. Which equilibria are internally consistent? How much effort is required, in network-weighted terms, to steer the system from one structural equilibrium to another? Which components are the true system-wide levers once all chains of indirect influence are counted? In what temporal sequence do the components of the system adjust during the transition? And how does the network at a given structural equilibrium respond to a unit push on a single component, independently of any specified destination?\\

\noindent Approaches to such assessments are distributed across disciplines. Input-output (IO) and computable general equilibrium models (CGE), for example, encode production and demand interdependencies and their response to perturbation \citep[e.g.][]{miller2009input, ShovenWhalley1984,BOHRINGER2008574}; integrated assessment models (IAM) combine energy, climate, and economic modules to project quantitative transition pathways \citep[e.g.][]{Weyant2017,LIU2026287}; energy system optimisation models (ESOM) represent detailed energy technologies, flows, and resource constraints to identify least-cost system configurations \citep[e.g.][]{HORSCH2018207,Weinand2026}; and system dynamics and agent-based frameworks capture feedback mechanisms and emergent adjustment dynamics \citep[e.g.][]{Sterman2000,HASUMI2025238,POLEDNA2023104306}. Causal mapping methods encode expert-elicited influence relationships \citep[e.g.][]{Godet2001, Kosko1986}; network-theoretic approaches identify structural leverage points and system controllability \citep{LiuSlotineBarabasi2011}; and econometric methods detect and date structural breaks in observed time series \citep{BaiPerron1998,Hansen}. No single approach, however, fully addresses all five structural questions posed above, and the large-scale quantitative frameworks among them depend on externally supplied scenario inputs whose coherence they do not verify.\\

\noindent The large-scale quantitative frameworks in this landscape (e.g. CGE, ESOM, and IAM models) require some form of external scenario guidance for their assessments and for this purpose can draw on internally consistent socio-technical scenarios derived from an expert-elicited qualitative interdependency structure. Such scenarios specify the policy environment, the technology mix, the governance arrangements, and the demand conditions that together constitute the structural setting each model evaluates. Those scenario assumptions must be internally coherent. Yet standard scenario construction, which assembles qualitative components from independently specified storyline elements, provides no formal mechanism for verifying such coherence. Within such socio-technical scenarios, the further structural questions of transition effort, indirect leverage, adjustment sequencing, and shock response remain unanswered across the methods surveyed above. The first of these gaps motivates the use of structured scenario methods that derive internally consistent futures as equilibria of a system of interdependencies rather than assembling them from independent storyline components.\\

\noindent Cross-impact balance (CIB) analysis is a structured elicitation method that encodes this interdependency structure in a matrix of pairwise influence scores across the dimensions of a socio-technical system to derive consistent scenarios \citep{Weimer-Jehle2006}. In this context, a consistent scenario is a structural equilibrium of the encoded interdependency network, in which every component is in the state most strongly supported by all other components simultaneously. Such equilibria are the fixed-point attractors of the cross-impact network, and the elicited cross-impact matrix (CIM) records the structural coupling among components \citep{Weimer-Jehle2006}. Standard CIB thus provides an equilibrium catalogue, a complete set of internally consistent configurations, answering the first question posed above. These internally consistent configurations, together with the socio-technical storylines that accompany them, are also the qualitative inputs upon which quantitative assessment frameworks such as ESOM, CGE, and IAM analyses can draw \citep{schweizer2020reflections,weimer2020socio,WEIMERJEHLE2016956}. The remaining four questions, however, lie outside the scope of standard CIB.\\

\noindent Several strands of work have attempted to extend CIB beyond this equilibrium catalogue. One strand concerns temporal dynamics, specifically how the system moves between attractors and what determines pathway distributions. \citet{Weimer-Jehle2006} proposes relaxing the deterministic succession rule by selecting the next descriptor state stochastically. \citet{Schweizer2023} formalises this approach as a Markov chain with probabilistic diagnostics and information-theoretic measures, and \citet{Campfens2025} integrates feedback loop analysis with succession to identify intervention points capable of initiating cascades towards desired attractors. A parallel strand addresses uncertainty about the CIM itself rather than the update rule. \citet{Salo2022} and \citet{Roponen2024} reframe CIB in fully probabilistic terms, synthesising cross-impact judgements into Bayesian networks, whilst \citet{Ross2026} implements a score-perturbation Monte Carlo workflow within a unified pathway framework.\\

\noindent The succession-based extensions face two structural limitations that the present framework addresses directly. The succession operator advances every descriptor simultaneously at each period, imposing a uniform response timescale across all descriptors. The apparent ordering of descriptor movements in succession-based trajectories is consequently an artefact of the protocol's uniform period length, not a property of the cross-impact network. A case in point is the Markov-chain formalisation of \citet{Schweizer2023}, the most rigorous succession-based treatment to date. As \citet{Schweizer2023} find, stochastic update rules gain information about transition pathways but may lose information about system attractors in the process. More fundamentally, applying succession dynamics to the CIM implicitly treats expert-elicited qualitative influence scores as though they were empirically grounded parameters governing the causal propagation of structural change. Standard CIB elicitation establishes pairwise influence directions and approximate relative magnitudes \citep{Weimer-Jehle2006,Weimer-Jehle2023}. However, it does not produce the identified coefficients of a state-space model. Using it as though it did conflates elicitation with identification \citep{Ross2026}. The score-perturbation strand, which propagates uncertainty through the matrix entries rather than running the succession operator, does not add this specific conflation, though the epistemological limitations inherent in expert-elicited CIM entries apply to all CIB-based approaches. It is this strand that the present framework draws on for robustness assessment.\\

\noindent Against this background, all four remaining questions posed in the opening paragraph remain unanswered by either standard or dynamic CIB. The indirect-leverage question warrants particular attention. Unlike transition effort, adjustment sequencing, and shock response (which the equilibrium catalogue plainly cannot supply), the raw cross-impact scores appear to provide leverage information directly, and do so misleadingly. Raw cross-impact scores encode only first-order pairwise links, so feedback loops in which sustained pressure on a descriptor ultimately reverses its own effect through the wider network are entirely invisible in the elicited matrix, and a monitoring framework built on those scores alone will misread such transient reversals as evidence that a transition is failing rather than as intrinsic features of the network's adjustment dynamics.\\

\noindent This paper introduces the CIB-LRT framework. Its point of departure is a structural observation that has not, as far as can be established, previously been made: the CIB drift matrix and the Leontief input-output technology matrix are algebraically identical in form. When a consistent scenario is used as a reference point and the CIM is linearised around it, the resulting drift matrix $M = W - I_N$ (where $W$ encodes the pairwise cross-impact scores at the consistent scenario and $I_N$ is the $N \times N$ identity matrix) is algebraically identical to the standard dynamic Leontief drift matrix. This isomorphism is not a modelling choice but a structural property of the two frameworks, and it makes the full Green-Kubo Linear Response Theory (LRT) machinery of \citet{Klimek2019} directly applicable to CIB. LRT provides in this context an analytic link between the microscopic equilibrium fluctuations of the system and its macroscopic out-of-equilibrium response to exogenous changes in that system \citep{Klimek2019}, and delivers four analytical objects. These objects are the perturbation budget, the Type I cross-impact multiplier, the impulse response function (IRF), and the unit-impulse shock profile. All four are derived in closed form from the CIM at a fixed consistent scenario, without any succession step, as direct consequences of this isomorphism.\footnote{The primary economic contribution of the LRT framework is due to \citet{Klimek2019}. \citet{Raseta}, available as a discussion paper, were the first to derive the closed-form expressions for the susceptibility, impulse response function, and multiplier in the IO economics setting, and confirmed that these reproduce the simulation-based results of \citet{Klimek2019}. The closed-form expressions used in the present paper are those first established by \citeauthor{Raseta}; Section~\ref{sec:S6} of the Supplementary Materials provides self-contained proofs adapted to the CIB setting that do not require \citeauthor{Raseta} to be formally published. The IO-LRT framework has since been applied in a study to assess economic resilience and susceptibility to extreme events and disruptions \citep{RossEtAl2025}.} Because each descriptor relaxes at a rate determined by its role in the cross-impact network, the ordering of responses is a structural property of the CIM rather than an artefact of the succession protocol's uniform period length. The Green-Kubo relations \citep{Green1954,Kubo1957} and the Ornstein-Uhlenbeck framework are standard tools of mathematical physics and statistical mechanics, brought into the economics literature via econophysics (e.g. \citet{Klimek2019}).\\

\noindent The CIB-LRT framework takes the critique of succession-based dynamics as a design constraint rather than a limitation. By deriving all four analytical objects without running a single succession step, the framework avoids the conflation of elicitation with identification that succession-based dynamics entail. It does not, however, resolve the weaker empirical grounding of a small-sample expert-elicited CIM relative to e.g. interactions estimated from large-scale transaction data. For robustness assessment, the framework operates in the score-perturbation mode, and the perturbation budget is shown to be a robust structural property of the cross-impact network, remaining stable under score perturbations at the uncertainty levels identified by \citet{Salo2022}.\\

\noindent With these properties established, the framework is applied empirically to a 15-descriptor energy-transition CIM from \citet{RossAGandAM2026}, which encodes the structural interdependencies among policy, technology, infrastructure, and demand-side dimensions of a low-carbon transition. However, the framework is domain-agnostic. The analytical objects are derived solely from the CIM, require no additional data beyond what a standard CIB elicitation already produces, and transfer without modification to any domain in which expert-elicited pairwise influence scores encode structural interdependencies, including energy \citep[e.g.][]{VOGELE2017937}, water(conflict) \citep[e.g.][]{kosow2024uncharted}, transport and mobility \citep[e.g.][]{TORI2023103160}, and climate-adaptation governance \citep[e.g.][]{schweizer2020reflections}. The energy-transition application is illustrative; the method is general.\\

\noindent This paper is organised as follows. Section~\ref{sec:methods} sets out the formal model, deriving the four analytical objects from the CIB cross-impact structure. Section~\ref{sec:application} describes the energy-transition application and situates each descriptor within the structural economics of low-carbon transition. Section~\ref{sec:results} presents the results, with the per-descriptor IRF as the primary focus. Section~\ref{sec:discussion} discusses limitations and directions for further development. The Supplementary Materials provide full sensitivity analyses.\\

\section{Methods}
\label{sec:methods}

\noindent The following subsections set out the formal model.\footnote{A related but distinct set of tools for extending CIB is provided by dynamic pathway approaches such as the framework of \citet{Ross2026}. These approaches run the CIB succession operator repeatedly under shocks, alternative regimes, and uncertain cross-impact matrices to reveal the distribution of attractor endpoints across runs and to quantify disequilibrium along pathways. The LRT objects described in this section are different in kind, since no amount of pathway simulation produces a perturbation budget, a cross-impact multiplier, an impulse response function, or a unit-impulse shock profile. The CIB-LRT framework and dynamic pathway methods therefore address different questions about the same cross-impact matrix, and a thorough CIB analysis can draw on both.} To make the concepts more accessible, a simple analogy is offered. \citet{Klimek2019} introduce their linear response framework with a house-of-cards analogy, picturing an economy as a house in which each sector (card) is held in place by the surrounding structure. The same picture applies directly to CIB. A consistent scenario is a configuration in which every descriptor sits in the state most strongly supported by all other descriptors simultaneously. In this arrangement, the house stands because every card (descriptor) is positioned precisely where the surrounding structure favours it. The CIM records the directional strength of each card's dependence on every other. Gently push one card and two distinct things can be read off. The first is how far that card itself moves. A card that the surrounding structure holds firmly barely shifts and is costly to displace, whereas one that the structure only weakly supports moves easily. The second is how far the rest of the house moves in response. A card on which many others depend is a structural load-bearer, so pushing it sends the widest cascade of adjustments through the surrounding cards, even if that card itself hardly moves. These two readings, a card's own resistance to being displaced and its leverage over the others, together with the order in which the cards move once the push begins, are what the analytical objects below quantify.\\

\noindent From this encoding, four questions become answerable in closed form. The first, answered by the perturbation budget, asks how much sustained, directed effort is needed to reorganise the house from one consistent configuration into a specific alternative. The second, answered by the cross-impact multiplier, asks which descriptors are the true structural load-bearers once all indirect chains of support are counted. The third, answered by the impulse response function, asks in what temporal sequence descriptors respond once the push begins, and whether any overshoot before settling. The fourth, answered by the unit-impulse shock profile, asks how a unit push on a single card (descriptor) propagates through the network at a given attractor, independently of any intended destination. Sections~\ref{sec:methods:cib}--\ref{sec:methods:irf} derive each object formally. Fuller accounts of the standard CIB method are given in \citep{Weimer-Jehle2006,Weimer-Jehle2023}. Full replication code is available at the repository identified in the Data Statement.\\

\subsection{CIB cross-impact structure}
\label{sec:methods:cib}

\noindent The CIB method is built on a CIM, a table of pairwise influence scores between descriptor states.\footnote{The impact scores $C_{ij}(k,l)$ are ordinal expert judgements, typically elicited on a scale from $-3$ to $+3$, and are summed across descriptors as if they carried interval properties they do not formally possess. The consistent scenarios identified by the argmax condition are therefore fixed points of a particular panel's collective judgement rather than empirically estimated equilibria. Uniform positive rescalings of all scores leave the consistent scenarios unchanged (this rescaling, labelled IO-3 in the invariance-operation numbering of \citet{Weimer-Jehle2006}, multiplies every off-diagonal cell by the same positive scalar $\alpha > 0$, leaving the argmax invariant), so qualitative findings are robust to that class of transformation. The same invariance is exploited in Section~\ref{sec:methods:ou} to rescale $W$ so that the drift matrix satisfies the stability condition of the LRT framework. Non-uniform monotone rescalings of the ordinal scale would alter score ratios and thereby change the consistent scenarios themselves, making a controlled robustness test for that class ill-defined within the standard CIB framework.} Let there be $N$ descriptors $d_1, \ldots, d_N$, each with a finite ordered set of states. A scenario $\mathbf{z} = (z_1, \ldots, z_N)$ assigns one state per descriptor. The CIB matrix supplies impact scores $C_{ij}(k, l)$ recording the influence exerted by state $k$ of descriptor $i$ on state $l$ of descriptor $j$ for all pairs $i \neq j$, with self-impacts excluded by convention. Given a scenario $\mathbf{z}$, the aggregate support received by state $l$ of descriptor $j$ from all other descriptors is

\begin{equation}
\theta_{j,l}(\mathbf{z}) = \sum_{i \neq j} C_{ij}(z_i, l),
\label{eq:impact}
\end{equation}

where $z_i$ is the state of descriptor $i$ in $\mathbf{z}$. A consistent scenario $\mathbf{z}^*$ is a structural equilibrium in which, for every descriptor $j$, the active state $z^*_j$ attains the maximum aggregate support,

\begin{equation}
z^*_j \in \operatorname{argmax}_{l} \;\theta_{j,l}(\mathbf{z}^*),
\label{eq:consistency}
\end{equation}

so that every descriptor is already in the state most strongly supported by the rest of the configuration. Consistent scenarios are the fixed-point attractors of the CIB succession operator and constitute the structural equilibrium landscape of the system. The LRT framework is defined at a fixed consistent scenario $\mathbf{z}^*$ and is local to that attractor. The effective CIM and drift matrix are derived in the subsections that follow.\\

\subsection{Effective cross-impact matrix and drift matrix}
\label{sec:methods:linearise}

\noindent To apply LRT, the discrete CIB state space is embedded in $\mathbb{R}^N$ via a centred encoding whereby descriptor $j$'s state index is shifted so that the consistent state $z^*_j$ maps to zero. A continuous activation variable $x_j(t) \in \mathbb{R}$ represents the deviation of descriptor $j$ from its consistent value. The effective CIM $W \in \mathbb{R}^{N \times N}$ is then extracted from the CIB matrix at $\mathbf{z}^*$ by reading off the entries corresponding to the reference scenario:

\begin{equation}
W_{ji} = C_{ij}(z^*_i, z^*_j) \quad (i \neq j), \qquad W_{jj} = 0.
\label{eq:W}
\end{equation}

\noindent The entry $W_{ji}$ records the direct cross-impact of descriptor $i$ on descriptor $j$ at the consistent scenario $\mathbf{z}^*$. The zero diagonal reflects the CIB convention that self-impacts are excluded. $W$ is in general asymmetric, since the influence of $i$ on $j$ and the influence of $j$ on $i$ are independently elicited and need not be equal. The drift matrix is defined as

\begin{equation}
M = W - I_N,
\label{eq:M}
\end{equation}

where $I_N$ is the $N \times N$ identity matrix. Subtracting $I_N$ adds a uniform restoring force of $-1$ to every descriptor's own equation, so that any descriptor displaced from its consistent state is pulled back with a strength proportional to the displacement. The off-diagonal entries of $M$ are the cross-impact scores (recording how one descriptor influences another), whilst the diagonal entries are uniformly $-1$ because the zero diagonal of $W$ means each descriptor exerts no direct self-impact in the CIB framework.\footnote{The form $M = W - I_N$ is algebraically identical to the IO drift matrix of \citet{Klimek2019}, where the hollow coupling matrix is denoted $A_0$. Substituting $W$ for $A_0$ makes all their results directly applicable to CIB without modification.}\\

\subsection{The CIB Ornstein-Uhlenbeck process}
\label{sec:methods:ou}

\noindent The activation vector $\mathbf{x}(t) = (x_1(t), \ldots, x_N(t))^\top$ is modelled as the solution of the stochastic differential equation (SDE):

\begin{equation}
\mathrm{d}\mathbf{x}(t) = M\,\mathbf{x}(t)\,\mathrm{d}t + \mathrm{d}B(t),
\label{eq:SDE}
\end{equation}

where $B(t)$ is an $N$-dimensional standard Brownian motion with identity covariance ($\mathrm{d}B(t)\,\mathrm{d}B(t)^\top = I_N\,\mathrm{d}t$). This is the CIB \citet{uhlenbeck1930theory} (OU) process. Its component form is

\begin{equation}
\mathrm{d}x_j(t) = \Bigl(\sum_{i \neq j} W_{ji}\, x_i(t) - x_j(t)\Bigr)\mathrm{d}t + \mathrm{d}B_j(t).
\label{eq:SDE_comp}
\end{equation}

\noindent The drift consists of a cross-impact term $\sum_{i \neq j} W_{ji} x_i(t)$, which propagates activation pressures from all other descriptors, and a restoring force $-x_j(t)$, which pulls descriptor $j$ back towards its consistent state. At every moment each descriptor is simultaneously pushed by the cumulative network pressure and pulled back towards equilibrium by its own restoring force. The balance between these two forces, mediated by the cross-impact network, determines whether the system is stable and how quickly it recovers from perturbation. The SDE~\eqref{eq:SDE_comp} is the CIB analogue of the IO economy SDE studied in \citet{Klimek2019}.\\ 

\noindent This is adopted here as a modelling ansatz rather than derived from first principles of CIB dynamics, given that consistent scenarios are discrete qualitative fixed points, not equilibrium states of a fluctuating physical or economic system, and the natural equilibrium fluctuations that justify the Green-Kubo formalism through the fluctuation-dissipation theorem \citep{callen1951irreversibility} are not present in a CIB attractor. The framework's value rests on the structural isomorphism between the CIB and IO drift matrices, which allows the key analytical results of \citet{Klimek2019} to transfer in closed form, and on the tractability of the resulting analytical objects.\\

\noindent The framework requires that all eigenvalues of $M$ have strictly negative real parts (the stability condition).\footnote{The stability condition is sufficient for all inversions required by the framework. When all $\operatorname{Re}(\lambda) < 0$, every eigenvalue of $\exp(M^\top t_1)$ lies strictly inside the unit circle, so none equals $1$ and $(\exp(M^\top t_1) - I_N)$ is non-singular (equations~\eqref{eq:susceptibility} and~\eqref{eq:eps} below). Stability also implies that $0$ is not an eigenvalue of $M$, so $M^{-\top}$ and hence $\Lambda = -M^{-\top}$ (the Type I cross-impact multiplier, Section~\ref{sec:methods:multiplier}) are likewise well-defined. No additional conditions are required beyond those already imposed by the IO-3 rescaling.} When this holds, the OU process admits a unique stationary distribution, which is a multivariate Gaussian with mean zero and covariance matrix $\Sigma$ satisfying $M\Sigma + \Sigma M^\top + I_N = 0$ (\citealt{Risken1989}). In practice, CIB matrices often violate this condition because the spectral radius $\rho(W)$ exceeds unity.\footnote{The symbol $\rho$ serves three distinct roles in this paper, following the notation of \citet{Klimek2019}: $\rho(W)$ denotes the spectral radius of matrix $W$; $\rho_{\mathrm{target}}$ denotes a chosen target value for that spectral radius; and $\rho(t_1)$ (introduced in Section~\ref{sec:methods:susceptibility}) denotes the susceptibility matrix as a function of the calibration horizon. All three uses are always distinguishable by their argument.} The IO-3 invariance of the CIB method (uniform positive rescaling of all cross-impact scores, which leaves the consistent scenarios unchanged) provides a natural remedy: replacing $W$ by $\alpha W$ gives $M_\alpha = \alpha W - I_N$, which is stable for any $\alpha < 1/\rho(W)$.\footnote{In practice, $\alpha$ is chosen to set a target spectral radius $\rho_{\mathrm{target}} \in (0,1)$, giving $\alpha = \rho_{\mathrm{target}}/\rho(W)$, an analyst choice with no uniquely correct value. The eigenvectors of $M_\alpha = \alpha W - I_N$ are identical to those of $W$ for any $\alpha > 0$: the substitution $W \mapsto \alpha W$ scales all eigenvalues by $\alpha$, and subtracting $I_N$ shifts each by $-1$, but neither operation rotates the eigenvector basis. It follows that the ranking of descriptors by their contribution to each eigenmode (and therefore the structural ordering of IRF responses) is invariant to the choice of $\rho_{\mathrm{target}}$, a result confirmed numerically across the full range $\rho_{\mathrm{target}} \in [0.50, 0.99]$ examined in the Supplementary Materials. When the dominant eigenvalue of $W$ is real, the slowest relaxation timescale of $M_\alpha$ equals $1/(1-\rho_{\mathrm{target}})$ exactly. When it is complex, the slowest timescale is shorter and must be read from the eigenspectrum of $M_\alpha$ directly.} In what follows, $M$ denotes the stable drift matrix; wherever $M = W - I_N$ is not stable, this means $M_\alpha$ with $\alpha$ chosen as described above.\\

\noindent The OU process~\eqref{eq:SDE} is dimensionless. Its running time $t$, the response lag $\tau$ that indexes the impulse response functions (Section~\ref{sec:methods:irf}), and the calibration horizon $t_1$ (Section~\ref{sec:methods:susceptibility}) all measure the same internal model clock, one unit of which ($\tau = 1$, equivalently $t = 1$) corresponds to neither one year nor any other calendar period. Model time is therefore not directly interpretable as a real-world duration. This issue is not specific to CIB-LRT; it applies equally to all CIB-based approaches. Attaching calendar meaning requires an external argument linking the model clock to observed transition speeds in the domain being studied, and such an argument would need to be grounded in empirical evidence outside the CIB elicitation process. One partial remedy is available within the existing CIB epistemological framework. The expert panel could be asked to supply a single temporal calibration anchor, specifically a collective judgement on the approximate real-world duration before the fastest-responding descriptor would show a meaningful response to a shock. That anchor would fix the location of the first IRF peak in calendar time, rescaling the model-time axis without altering the internally derived ordering of descriptor responses, which remains a property of the matrix eigenstructure. The resulting calendar axis carries the same epistemological status as the underlying cross-impact scores, being expert-elicited, qualitative, and conditional on panel composition. It does not constitute empirical identification of propagation speeds.\\

\subsection{Susceptibility matrix}
\label{sec:methods:susceptibility}

\noindent The susceptibility matrix $\rho(t_1)$ is the Green-Kubo transport coefficient associated with the OU process~\eqref{eq:SDE}, and is the generating object from which the four analytical objects are obtained rather than one of the four itself. It records, in matrix form, how an external constant forcing $\widetilde{\boldsymbol{\varepsilon}} \in \mathbb{R}^N$ over the interval $[0, t_1]$ translates into a cumulative activation pattern across descriptors. Its inverse applied to a target displacement $\Delta\mathbf{z} \in \mathbb{R}^N$ (the vector of state-index differences between two consistent scenarios, formalised in Section~\ref{sec:methods:budget}) defines the implied perturbation vector, which anchors the perturbation budget and impulse response function (Section~\ref{sec:methods:budget}). Its long-run limit as $t_1 \to \infty$ yields the Type I cross-impact multiplier (Section~\ref{sec:methods:multiplier}). The susceptibility matrix at horizon $t_1 > 0$ evaluates in closed form as

\begin{equation}
\rho(t_1) = M^{-\top}\bigl(\exp(M^\top t_1) - I_N\bigr),
\label{eq:susceptibility}
\end{equation}

where $M^{-\top} = (M^\top)^{-1}$ denotes the inverse transpose and $\exp(\cdot)$ denotes the matrix exponential.\footnote{A note on convention is required for readers following the derivation. The Green-Kubo construction defines $\rho(t_1)$ as an integrated equilibrium cross-correlation function. For symmetric drift this coincides with the SDE-derived mean response of $\mathbf{x}$ to constant forcing $\widetilde{\boldsymbol{\varepsilon}}$, namely $\mathbb{E}[\mathbf{x}(t_1)\mid\widetilde{\boldsymbol{\varepsilon}}] = M^{-1}(\exp(Mt_1)-I_N)\widetilde{\boldsymbol{\varepsilon}}$. For the asymmetric drift matrices that arise from CIB and from IO economics, the two objects differ by a transpose, in that the SDE response is the transpose of $\rho(t_1)\widetilde{\boldsymbol{\varepsilon}}$, not equal to it. The framework here, following \citet{Klimek2019}, takes the Green-Kubo susceptibility as the defining object, so the identity $\rho(t_1)\widetilde{\boldsymbol{\varepsilon}} = \Delta\mathbf{z}$ used in Section~\ref{sec:methods:budget} is a matrix-level relationship between implied perturbation and target displacement, not the literal time-$t_1$ mean trajectory under~\eqref{eq:SDE}. The perturbation budget, the multiplier, and the impulse response function are Green-Kubo-convention quantities and will in general differ from their SDE-convention counterparts when $M$ is asymmetric. The matrix-level interpretation of individual entries used throughout this paper, namely $\rho_{ki}(t_1) \tilde\varepsilon_i$ as the contribution of forcing on descriptor $i$ to the target displacement of descriptor $k$, is internally consistent.} The derivation proceeds through the Green-Kubo integral. The susceptibility is defined as $\rho(t_1) = \Sigma^{-1}\int_0^{t_1}\langle\mathbf{x}(0)\,\mathbf{x}(\tau)^\top\rangle_0\,\mathrm{d}\tau$ (where $\langle\cdot\rangle_0$ denotes the stationary expectation and $\tau \geq 0$ is the lag), and since $\langle\mathbf{x}(0)\mathbf{x}(\tau)^\top\rangle_0 = \Sigma\exp(M^\top\tau)$ for the OU process (Proposition~\ref{prop:crosscov} of the Supplementary Materials), substitution gives $\Sigma^{-1}\int_0^{t_1}\Sigma\exp(M^\top\tau)\,\mathrm{d}\tau = \int_0^{t_1}\exp(M^\top\tau)\,\mathrm{d}\tau = M^{-\top}(\exp(M^\top t_1) - I_N)$. The stationary covariance $\Sigma$ cancels exactly (the key technical fact enabling this application of the Green-Kubo formalism to CIB), leaving a result that depends only on the drift matrix $M$ and not on the noise structure of the OU embedding, so $\Sigma$ does not need to be estimated from data. Full proofs are given in Propositions~\ref{prop:crosscov} and~\ref{prop:susceptibility} of the Supplementary Materials.\\

\noindent The $(k,i)$ entry of $\rho(t_1)$ measures how much cumulative activation pressure accumulates on descriptor $k$ when descriptor $i$ is pushed steadily from outside the system for a period $t_1$, through all direct and indirect paths in the cross-impact network. A large entry $(k,i)$ means descriptor $i$ is an effective lever for reaching descriptor $k$, even if the direct cross-impact score $W_{ki}$ is small or zero. The standard CIB framework uses these direct, pairwise scores without aggregating their multi-step propagation through the network, and therefore does not compute this cumulative indirect influence.\\

\noindent The horizon $t_1$ is a free modelling parameter. The CIB framework assigns no physical time unit to $t$, so $t_1$ must be treated as a sensitivity parameter. At short horizons, $\rho(t_1) \approx t_1 I_N$, so each descriptor responds independently and cross-descriptor propagation is negligible. As $t_1 \to \infty$, $\rho(t_1)$ converges to the long-run multiplier $\Lambda$ (Section~\ref{sec:methods:multiplier}).\\

\subsection{Implied perturbation and perturbation budget}
\label{sec:methods:budget}

\noindent The term `budget' in the CIB literature typically refers to the combinatorial weight of an attractor's basin of attraction \citep{Weimer-Jehle2006}, specifically the count of starting scenarios that flow to a given consistent scenario under iterative succession. The perturbation budget defined here is a distinct concept (the Euclidean norm of the implied perturbation vector $\widetilde{\boldsymbol{\varepsilon}}$), measuring network-weighted transition effort rather than basin size. It quantifies the magnitude of the sustained external forcing required to achieve a given scenario displacement through the cross-impact network, using the network's own propagation structure to weight the effort.\\

\noindent Given a target activation displacement $\Delta\mathbf{z} \in \mathbb{R}^N$ from the consistent scenario (the vector of state-index differences between the origin and target consistent scenarios), the implied perturbation vector is the unique constant forcing satisfying the Green-Kubo identity $\rho(t_1)\widetilde{\boldsymbol{\varepsilon}} = \Delta\mathbf{z}$ over the calibration window $[0, t_1]$:

\begin{equation}
\widetilde{\boldsymbol{\varepsilon}} = \rho(t_1)^{-1}\,\Delta\mathbf{z} = \bigl(\exp(M^\top t_1) - I_N\bigr)^{-1} M^\top\,\Delta\mathbf{z}.
\label{eq:eps}
\end{equation}

\noindent Concretely, $\widetilde{\boldsymbol{\varepsilon}}$ specifies, for each descriptor, how hard the network must be pushed from outside to steer the system from its current scenario to the target. Descriptors for which the network is a natural amplifier require only a small push, whereas those that the network actively resists require a larger one. The Euclidean norm $\|\widetilde{\boldsymbol{\varepsilon}}\|$ collapses the full recipe into a single number (the perturbation budget), which is a network-weighted measure of the total effort required. Because the recipe is computed from the origin scenario's cross-impact structure rather than the target's, the effort of moving from A to B through network A differs from the effort of moving from B to A through network B. This asymmetry carries substantive information about the directional preferences of the cross-impact structure.\\

\noindent Two qualifications on the budget's interpretation should be stated explicitly. First, $\widetilde{\boldsymbol{\varepsilon}}$ is the unique constant forcing achieving $\Delta\mathbf{z}$ over $[0,t_1]$. The perturbation is constrained to be time-invariant. A time-varying forcing schedule could in principle achieve the same displacement with a different (possibly smaller) Euclidean norm, but the constant-forcing model is the analytically tractable benchmark consistent with the Green-Kubo framework and is not claimed to be the global minimum over all forcing shapes. The perturbation budget consequently provides an upper bound on the minimum achievable transition cost rather than the minimum itself, since an optimally shaped time-varying forcing would generally yield a lower value. Second, because the budget is derived from the linearisation at the origin scenario, it should be read as a network-weighted indicator of transition resistance, not as a precise engineering cost.\\

\noindent The perturbation budget is derived from the susceptibility matrix, which is the object from which \citet{Klimek2019} derive their economic resilience measure for IO systems. \citeauthor{Klimek2019} characterise resilience as the capacity of the IO network to absorb arriving shocks and recover, a property whose predictive value for GDP growth and recovery they demonstrate empirically. The budget uses the inverse of the susceptibility matrix applied to a specific target displacement, extending that foundation to the specific planning question of how much sustained forcing the cross-impact network requires for a given transition. A large budget reflects a network that provides little propagation support for the required descriptor changes, making the transition costly in network-weighted terms, but this is resistance to a specific directed change rather than a measure of general shock vulnerability.\\

\subsection{Cross-impact multiplier}
\label{sec:methods:multiplier}

\noindent As $t_1 \to \infty$, the susceptibility matrix converges to the long-run cross-impact multiplier (termed here the ``Type I cross-impact multiplier'')\footnote{Distinctly different from the ``cross-impact multiplier'' specified in \citet{Salo2022}. A formal proof of the limit is given in Corollary~\ref{cor:multiplier} of the Supplementary Materials.}:

\begin{equation}
\Lambda = \lim_{t_1 \to \infty} \rho(t_1) = -M^{-\top} = (I_N - W)^{-\top}.
\label{eq:multiplier}
\end{equation}

\noindent When $\rho(W) < 1$, which holds after IO-3 rescaling by construction, standard results give the convergent Neumann series $\Lambda = \sum_{n=0}^\infty (W^\top)^n$, which is the Leontief inverse of~$W^\top$ (equivalently, the transpose of the Leontief inverse of~$W$), and is the CIB analogue of the Type I IO multiplier \citep{miller2009input, Emonts-Holley02102021}, accumulating all chains of direct and indirect influence through the network (a structure formally analogous to \citealt{Katz1953} centrality in network science). Whilst in the non-negative IO setting each power $(W^\top)^k$ carries a clean interpretation as a distinct order of indirect requirements, in CIB the mixed-sign entries of $W$ mean individual powers have mixed signs and cannot be read as separate, additively interpretable rounds of positive indirect influence, so the series is best treated as an aggregate rather than decomposed term by term. The LRT derivation, arriving at $\Lambda$ as the long-run limit of the susceptibility integral (Section~\ref{sec:methods:susceptibility}), supplies a dynamical interpretation (cumulative activation response to sustained forcing) that the algebraic form alone does not provide.\footnote{The Type I cross-impact multiplier allows practitioners to rank descriptors by their true system-wide influence, accounting for all feedback chains rather than only direct cross-impacts. Conversely, a large direct impact that is attenuated in the multiplier signals feedback cancellation within the network. A further practitioner use is rapid identification of network insulators, namely descriptors that show near-zero response to a shock in descriptor~$i$ across the full time horizon and are therefore structurally decoupled from~$i$ within the network at that attractor.}\\

\noindent The $(k,i)$ entry of $\Lambda$ gives the total cumulative activation that accumulates on descriptor $k$ when descriptor $i$ is pushed steadily for a very long time, once all network effects have fully propagated. A large entry where the direct score $W_{ki}$ is small or zero signals a strong indirect connection through intermediary descriptors. A negative diagonal entry $\Lambda_{jj} < 0$ is particularly striking, as it means that pushing descriptor $j$ through the network ultimately reduces its own long-run activation, because the indirect feedback reverses the direction of the effect. Such feedback reversals are not quantifiable from the raw cross-impact scores and can only be identified through the multiplier analysis. Whilst the multiplier captures these indirect effects as a long-run aggregate, the impulse response function (Section~\ref{sec:methods:irf}) resolves the same propagation dynamically, tracing the temporal sequence of descriptor adjustments and identifying any transient overshoots along the way.\\

\subsection{Impulse response function}
\label{sec:methods:irf}

\noindent The CIB impulse response function (IRF) at lag $\tau \geq 0$ is the $N$-vector of descriptor activation pressures arising from the implied perturbation $\widetilde{\boldsymbol{\varepsilon}}$\footnote{This follows from the Green-Kubo response formula and the cross-covariance result by the same $\Sigma$-cancellation as the susceptibility; a proof is given in Corollary~\ref{cor:irf} of the Supplementary Materials.}:

\begin{equation}
R(\tau) = \exp(M^\top \tau)\,\widetilde{\boldsymbol{\varepsilon}}.
\label{eq:IRF}
\end{equation}

\noindent Here $R_j(\tau)$ is the activation pressure on descriptor $j$ at time $\tau$ after the push begins. Descriptors that are directly pushed or strongly coupled to the push through the network feel the pressure immediately (large $R_j$ at small $\tau$), while those reached only through long indirect chains respond later. The IRF integrates to the target displacement over $[0, t_1]$ (Corollary~\ref{cor:irf} of the Supplementary Materials), so the area under each descriptor's curve equals the state-index change prescribed for that descriptor. The time-resolved shape shows the order in which descriptors would need to move in a managed transition.\footnote{$R(\tau)$ is evaluated numerically by exploiting the semigroup property $\exp(M^\top(\tau + \Delta t)) = \exp(M^\top \Delta t) \cdot \exp(M^\top \tau)$. A single matrix exponential $\exp(M^\top \Delta t)$ is computed once, and successive time steps are obtained by repeated left-multiplication.}\\

\noindent The IRF can be decomposed into a set of exponentially decaying or oscillating-decaying patterns (eigenmodes), each with its own decay rate and, where the eigenvalues of $M$ are complex, its own oscillation frequency. Each eigenmode is a particular combination of descriptors that all fade at the same speed, and complex eigenmodes produce additional transient oscillations before they decay. Some modes decay quickly, with descriptors settling in a short period, while others decay slowly, leaving descriptors under pressure for a long time. A descriptor with a large weight in the slowest mode is the hardest to settle and the most persistent signal of an ongoing transition, whereas one with weight only in fast modes responds early and resolves quickly. This provides a temporal ranking of descriptor responses that no standard CIB analysis produces.\\

\noindent The fourth analytical object, the unit-impulse shock profile, arises when the perturbation is not calibrated to a particular scenario transition but instead consists of a unit push on a single descriptor $i$. Setting $\widetilde{\boldsymbol{\varepsilon}} = \mathbf{e}_i$ (the $i$-th standard basis vector in $\mathbb{R}^N$) yields the unit-impulse IRF for descriptor~$i$ (the closed form of that profile),

\begin{equation}
R^{(i)}(\tau) = \exp(M^\top \tau)\,\mathbf{e}_i,
\label{eq:unit_irf}
\end{equation}

which is simply the $i$-th column of the matrix exponential $\exp(M^\top\tau)$, evaluated at lag $\tau$. The integral of~\eqref{eq:unit_irf} over $[0, t_1]$ recovers the $i$-th column of the susceptibility matrix:

\begin{equation}
\int_0^{t_1} R^{(i)}(\tau)\,\mathrm{d}\tau
= M^{-\top}\bigl(\exp(M^\top t_1) - I_N\bigr)\mathbf{e}_i
= [\rho(t_1)]_{\cdot,\,i},
\label{eq:unit_irf_integral}
\end{equation}

where $[\rho(t_1)]_{\cdot,\,i}$ denotes the $i$-th column of the susceptibility matrix defined in~\eqref{eq:susceptibility}. Equations~\eqref{eq:unit_irf} and~\eqref{eq:unit_irf_integral} establish the precise relationship between the two analytical objects. The susceptibility matrix records the cumulative integral of all possible unit-impulse IRFs, and the unit-impulse IRF is the time-resolved version of a single column of the susceptibility matrix. The cross-scenario IRF introduced in~\eqref{eq:IRF} is a projection of this same structure onto the direction defined by the target displacement, given by $R(\tau) = R^{(1)}(\tau)\tilde\varepsilon_1 + \cdots + R^{(N)}(\tau)\tilde\varepsilon_N = \exp(M^\top\tau)\widetilde{\boldsymbol{\varepsilon}}$, where $\tilde\varepsilon_j$ denotes the $j$-th component of $\widetilde{\boldsymbol{\varepsilon}}$ (computed from~\eqref{eq:eps}).\\

\noindent The unit-impulse shock profile thus provides a descriptor-level decomposition of the cross-impact network's dynamic response at a given attractor, independent of any specified target transition. Because $W$, and hence $M$, is evaluated at a specific consistent scenario $\mathbf{z}^*$, the unit-impulse IRF~\eqref{eq:unit_irf} is attractor-specific, in that the same push on descriptor~$i$ propagates differently at each consistent scenario, and comparing the profiles across attractors yields a structural fingerprint of each scenario's local network sensitivity. Since the susceptibility matrix is the object from which \citet{Klimek2019} derive their economic resilience measure for IO systems, and the unit-impulse IRF is the time-resolved version of a column of that susceptibility matrix, the unit-impulse analysis is the CIB object most directly analogous to the resilience framing of \citet{Klimek2019}. A shock arrives at a single descriptor without any specified target transition, and the profile of $R^{(i)}(\tau)$ reveals how the cross-impact network at that attractor propagates and absorbs it.\\

\noindent The IRF integrates to give the cumulative displacement over any evaluation horizon $t_2$:

\begin{equation}
\begin{split}
\int_0^{t_2} R(\tau)\,\mathrm{d}\tau
&= \int_0^{t_2}\exp(M^\top\tau)\,\mathrm{d}\tau\cdot\widetilde{\boldsymbol{\varepsilon}} \\
&= M^{-\top}\bigl(\exp(M^\top t_2)-I_N\bigr)\widetilde{\boldsymbol{\varepsilon}}
= \rho(t_2)\,\rho(t_1)^{-1}\,\Delta\mathbf{z},
\end{split}
\label{eq:IRF_integral}
\end{equation}

where the second step uses the standard matrix-exponential integral identity and the third substitutes $\rho(t_2) = M^{-\top}(\exp(M^\top t_2)-I_N)$ and $\widetilde{\boldsymbol{\varepsilon}} = \rho(t_1)^{-1}\Delta\mathbf{z}$. At $t_2 = t_1$ this equals $\Delta\mathbf{z}$ by construction (the perturbation delivers exactly the target displacement over the calibration horizon). For $t_2 < t_1$ it gives the fraction of the transition completed at an earlier check-point, and for $t_2 > t_1$ it gives the extrapolated overshoot.\\

\section{Application}
\label{sec:application}

\noindent The framework is applied to a 15-descriptor energy-transition CIM from \citet{RossAGandAM2026}. The CIB-LRT framework is applicable to any cross-impact matrix that admits consistent scenarios. The energy-transition matrix provides a concrete, empirically relevant example that demonstrates the range of structural insights the framework offers to CIB practitioners. The matrix encodes the structural interdependencies among the principal dimensions of a low-carbon transition, spanning regulatory, technological, infrastructural, financial, and demand-side determinants of decarbonisation, with three ordered states per descriptor. Five consistent scenarios are identified using PyCIB \citep{RossAG2026pycib}, denoted sc01 to sc05 and summarised in Table~\ref{tab:scenarios}. The five scenarios fall into three structural groupings, comprising two fossil-locked configurations (sc01, sc02), two intermediate configurations (sc03, sc04), and one net-zero aligned configuration (sc05). Within each of the two near-degenerate pairs, sc01/sc02 and sc03/sc04, the scenarios differ on a single descriptor (Hamming distance $H = 1$; \citealt{hamming1950error}; see also \citealt{Ross2026}) and share nearly identical cross-impact structures.\\

\begin{table}[htbp]
\centering
\caption{Consistent scenarios for the energy-transition case study.}
\label{tab:scenarios}
\footnotesize
\begin{tabular}{llp{8.0cm}}
\toprule
Label & Broad character & Key features \\
\midrule
sc01 & Locked-in-fossil & Low Policy Stringency, high Energy Security Pressure, low Renewables Deployment, high Technology Costs. Land Use Conflict: Low. \\[3pt]
sc02 & Locked-in-fossil & As sc01; differs only on Land Use Conflict (Medium). \\[3pt]
sc03 & Intermediate & Descriptors predominantly at medium ordinal levels. Technology Costs: Low, Permitting Pace: Medium. \\[3pt]
sc04 & Intermediate & As sc03; differs only on Permitting Pace (Fast). \\[3pt]
sc05 & Net-zero aligned & High Policy Stringency, high Renewables Deployment, major Hydrogen Role, low Technology Costs; Industrial Energy Demand at medium. \\
\bottomrule
\end{tabular}
\smallskip\par\noindent
{\footnotesize Source: \citet{RossAGandAM2026}. Consistent scenarios identified using PyCIB \citep{RossAG2026pycib}.}
\end{table}

\noindent Three free parameters govern the implementation. These are the IO-3 target spectral radius $\rho_{\mathrm{target}} = 0.9$, the calibration horizon $t_1 = 1.0$ model units, and the Monte Carlo noise scale $\sigma = 0.5$ ordinal units. Because the spectral radius of the raw effective CIM exceeds unity at every attractor (values in the range 10.3--17.1 across the five scenarios), IO-3 rescaling is applied at each attractor, with scenario-specific rescaling factor $\alpha$ in the range 0.053--0.087. All LRT objects are computed from the rescaled matrix. The sensitivity of all three parameter choices is examined in Section~\ref{sec:results:sensitivity} and the Supplementary Materials. Uncertainty ribbons are constructed by Monte Carlo simulation over the cross-impact scores. For each of 10,000 realisations, independent Gaussian noise ($\sigma = 0.5$) is added to every off-diagonal entry of $W$, IO-3 rescaling is re-applied, and the IRF is recomputed. Realisations in which $M$ is unstable or the linear system is singular are discarded. All 10,000 realisations are accepted in this application. The 5th--95th percentile band across accepted realisations forms the shaded confidence ribbon in each figure.\\

\section{Results}
\label{sec:results}

\noindent The results presented in this section are not intended as an exhaustive empirical analysis of the energy transition case study, but as a selective illustration of the structural insights that the CIB-LRT framework makes accessible to CIB practitioners. The aim is to demonstrate the range of analytical questions that the framework permits, none of which standard CIB can address. The results are organised in four parts. Section~\ref{sec:results:irf} presents the per-descriptor impulse response functions for the sc05$\to$sc03 transition. Section~\ref{sec:results:shock} presents the unit-impulse shock analysis, which characterises the network-propagated response to a unit push on a single descriptor at a fixed attractor. Section~\ref{sec:results:ensemble} examines all twenty directed transition pairs through the IRF norm ensemble and perturbation budgets. Section~\ref{sec:results:susceptibility} illustrates how indirect influence accumulates as the calibration horizon lengthens. Sensitivity and robustness diagnostics are summarised in Section~\ref{sec:results:sensitivity} and detailed in the Supplementary Materials.\\

\subsection{Per-descriptor impulse response trajectories}
\label{sec:results:irf}

\noindent Figure~\ref{fig:irf_sc05_sc03} shows the impulse response functions for the sc05$\to$sc03 transition (net-zero to intermediate). This transition is selected as the primary illustration because the sc05 attractor exhibits the most analytically distinctive eigenstructure of the five consistent scenarios, including negative cross-impact multiplier diagonals that produce the sign reversal discussed below, making it the most revealing anchor for demonstrating the full range of IRF features. Structurally, the system retreats from a fully decarbonised configuration where the principal dimensions of decarbonisation are mutually reinforcing, moving towards an intermediate state in which most descriptors settle at medium levels. Each line traces the activation pressure on one of the fifteen descriptors from $\tau = 0$ until the system has settled. All $\tau$ values are dimensionless model units, with no mapping to calendar years available without an external calibration argument that CIB does not provide. The solid lines use the agreed cross-impact scores. The shaded ribbons span the 5th--95th percentile of 10,000 Monte Carlo realisations.\\

\noindent One formal constraint governs all fifteen descriptor trajectories. The area under each descriptor's curve over the calibration horizon $[0, t_1]$ must equal the total state-index shift prescribed for that descriptor ($\Delta z_j$), so a descriptor that does not change between the two scenarios has zero net area, whilst one that moves two ordinal levels has twice the area of one that moves one level. It is the shape of the trajectory (its timing, direction, and any overshoots) that standard CIB cannot produce and that the impulse response supplies.\\

\noindent Two features of Figure~\ref{fig:irf_sc05_sc03} merit explicit explanation. First, why do some curves start above zero and others below? The sign of each curve's starting value gives the direction of the initial push the cross-impact network must apply to that descriptor to drive the transition to the target attractor. Crucially, this is not determined simply by whether the descriptor increases or decreases between the two scenarios, because indirect network paths can route the required adjustment in directions that look counterintuitive from the raw state differences alone. A descriptor whose net state change is zero may still start away from zero, if the network routes pressure through it as a transient intermediary. Its integrated area will nonetheless sum to zero. Second, why do all fifteen trajectories converge to zero? Zero activation pressure means the transition is complete. The network has absorbed the perturbation and settled at the new attractor. Convergence is guaranteed by the stability of the cross-impact structure, which ensures that any initial push must eventually dissipate. Zero is therefore not an arbitrary baseline but the resting point of a stable network.\\

\noindent For most descriptors, the sc05$\to$sc03 transition produces a smooth, monotone decay from the initial activation pressure towards zero. This behaviour reflects the eigenmode structure of sc05, which decomposes into one slow mode, with a relaxation timescale of ten model units set by the IO-3 rescaling target, and fourteen faster modes that resolve within a few model units. Slow-settling descriptors carry a gently decaying tail across the full horizon and are the natural candidates for ongoing structural monitoring, as their persistence reflects the CIM's eigenstructure rather than any particular parameter choice.\\

\begin{figure}[!ht]
    \centerline{%
    \includegraphics[keepaspectratio]{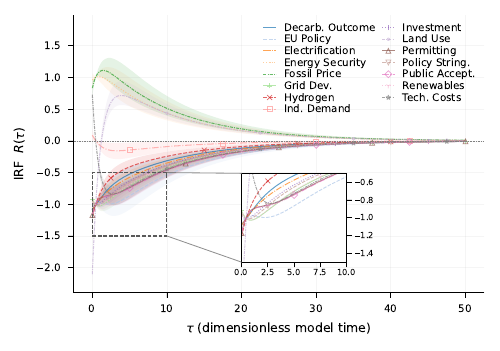}
    }
    \caption{Descriptor adjustment trajectories for the sc05$\to$sc03 transition. Per-descriptor impulse response functions $R_j(\tau) = [\exp(M^\top\tau)\widetilde{\boldsymbol{\varepsilon}}]_j$ for the transition sc05$\to$sc03 (net-zero to intermediate, Permitting Pace Medium). Hamming distance 13; perturbation budget $\|\widetilde{\boldsymbol{\varepsilon}}\| = 4.22$. Each line corresponds to one descriptor, giving the deterministic IRF from the agreed cross-impact scores. Shaded ribbons span the 5th--95th percentile of $N = 10{,}000$ Monte Carlo realisations for all fifteen descriptors (noise $\sigma = 0.5$ per off-diagonal CIM entry). Calibration horizon $t_1 = 1.0$; time axis in dimensionless model units (not calendar time). The area under each descriptor's curve over $[0, t_1]$ equals the target state-index displacement $\Delta z_j$ by construction.}
    \label{fig:irf_sc05_sc03}
\end{figure}

\noindent The Energy Security Pressure curve provides the most striking result. Its trajectory begins positive, rises to a local peak, then reverses sign and decays from below before returning to zero. This reversal is a direct consequence of a feedback loop in the cross-impact structure at the net-zero attractor, captured by the negative diagonal entry of the Type I cross-impact multiplier ($-0.39$ at sc05). Within the OU model, sustained pressure on this descriptor winds back its own cumulative activation through indirect cross-impacts. Nothing in the raw CIM scores hints at this; only the IRF reveals it. In a managed retreat from the net-zero scenario, Energy Security Pressure would therefore first overshoot its target level before correcting, a transient that any monitoring framework should anticipate rather than misread as failure.\\

\noindent The Fossil Price Pressure curve is the second to rise above zero, but its mechanism differs fundamentally from the Energy Security case. Whereas Energy Security Pressure reverses sign through a self-reinforcing feedback loop, the Type I cross-impact multiplier diagonal for Fossil Price Pressure at sc05 is approximately zero ($0.003$, compared with $-0.39$ for Energy Security Pressure), confirming that no equivalent reversal is present. The positive trajectory has a structurally direct cause. Specifically, at the net-zero attractor, the Low state of Fossil Price Pressure is held in place not by intrinsic stability but by a coalition of reinforcing cross-impacts, including High Policy Stringency, Low Energy Security Pressure, and the Net-zero Decarbonisation Outcome, among others. The sc05$\to$sc03 transition dismantles this coalition simultaneously, as each member shifts away from its net-zero state and withdraws its reinforcing contribution.\\

\noindent The upward pulse in the Fossil Price Pressure IRF is therefore the compounded consequence of multiple structural supports being withdrawn at once rather than any single direct cause. The net-zero configuration is, in this sense, a house of cards in that it stands because all its components mutually support one another, and moving away from it means withdrawing that entire mutual support structure at once. Unlike a physical house of cards, the collapse is not instantaneous. The impulse response shows, step by step, how each component adjusts as the shared structure dissolves, a temporal resolution that standard CIB cannot provide. More detailed data on the multiplier diagonal entries for all fifteen descriptors at each attractor are given in Supplementary Table~\ref{tab:lambda_diag}.\\

\noindent The Monte Carlo ribbons are narrow throughout, confirming that the qualitative shape of every descriptor's trajectory, including the direction of the Energy Security Pressure reversal, is a structural invariant of the cross-impact network rather than an artefact of the specific numerical scores. The robustness of this finding is discussed in Section~\ref{sec:results:sensitivity}.\\

\subsection{Unit-impulse shock responses and attractor-specific sensitivity profiles}
\label{sec:results:shock}

\noindent The cross-scenario IRF in Section~\ref{sec:results:irf} characterises the system's response to a calibrated perturbation aimed at a specific scenario transition. Figure~\ref{fig:shock_sc05_policy} addresses a structurally distinct question. What is the network-propagated response when a unit push is applied to a single descriptor independently of any target transition? This is the unit-impulse IRF~\eqref{eq:unit_irf}, evaluated at the net-zero attractor (sc05) for $i = \text{Policy Stringency}$. The figure shows the response curves for the fourteen non-shocked descriptors, so that the full propagation pattern of a single-descriptor shock through the network is visible at once.\\

\begin{figure}[!ht]
    \centerline{%
    \includegraphics[keepaspectratio]{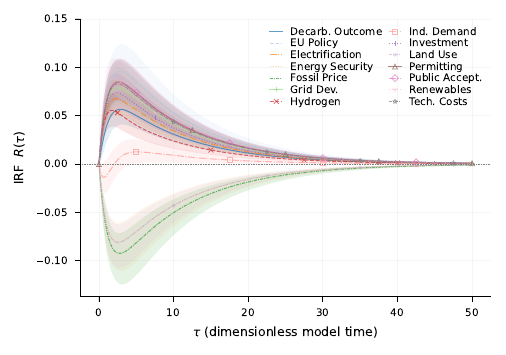}
    }
    \caption{Network shock propagation from a Policy Stringency impulse at sc05. Unit-impulse shock response $R^{(i)}(\tau) = \exp(M^\top\tau)\mathbf{e}_i$ at the net-zero attractor (sc05) for $i = \text{Policy Stringency}$. The shocked descriptor itself is not plotted; each line shows the deterministic response of one of the fourteen remaining descriptors. Shaded ribbons span the 5th--95th percentile of $N = 10{,}000$ Monte Carlo realisations ($\sigma = 0.5$ per off-diagonal CIM entry; all realisations accepted).}
    \label{fig:shock_sc05_policy}
\end{figure}

\noindent Four structural features of Figure~\ref{fig:shock_sc05_policy} are directly readable from the curves. First, three of the fourteen non-shocked descriptors respond negatively (Fossil Price Pressure, Land Use Conflict, and Energy Security Pressure), whilst eleven co-move positively, with the negative responses reflecting network routing at the net-zero attractor rather than direct links. Second, the strongest positive responders peak within the first few model units before decaying, reflecting the fast-mode structure of the net-zero drift matrix, in which most indirect propagation resolves quickly. Third, Industrial Energy Demand is almost entirely unresponsive, consistent with its peripheral structural position, and the unit-impulse IRF thus distinguishes central from peripheral descriptors in a way that raw CIM scores do not. Fourth, the Monte Carlo ribbons are narrow throughout, confirming that the qualitative propagation pattern is a robust structural property of the network.\\

\noindent The attractor-specificity of the unit-impulse profile is the analytically distinctive feature that no other CIB analysis provides. The same push on Policy Stringency applied at the fossil-locked attractor (sc01) would propagate through a structurally different cross-impact network, where Policy Stringency sits at its Low state, its reinforcing coalition is absent, and the dominant timescale is substantially shorter than at sc05. The unit-impulse IRF profile is therefore a structural fingerprint of each attractor, and comparing profiles across consistent scenarios reveals how the same external push encounters different structural contexts, with correspondingly different amplification patterns and recovery timescales. The house-of-cards analogy applies directly, since the same card plays a fundamentally different role depending on which house it stands in, and the unit-impulse IRF quantifies exactly how different those structural positions are. This attractor-conditional sensitivity analysis is not available through standard or dynamic CIB, because neither resolves the time-structured propagation of a descriptor-level shock within a given structural configuration.\\

\subsection{Scenario transition clusters, asymmetry, and perturbation budgets}
\label{sec:results:ensemble}

\noindent Figure~\ref{fig:irf_ensemble} collects the IRF norm $\|R(\tau)\|_2$ for all twenty directed transitions on a shared time axis. The norm aggregates activation pressure across all fifteen descriptors at each lag into a single number measuring how much total pressure the entire network is under at this moment in the transition. The twenty curves partition naturally into three groups. The three-cluster grouping is immediately interpretable through the house-of-cards picture. Moving between nearly identical scenarios is like moving a single card whilst the rest of the configuration stands intact, whereas moving between the fossil-locked and net-zero scenarios means dismantling one house entirely and building a structurally different one from scratch, which is why those transitions carry the largest budgets. This analogy describes the global structural contrast between configurations. However, the mathematical apparatus operates on the local linearisation at each endpoint, so budget values for large-displacement transitions are structural indicators at the reference attractor rather than estimates of the full nonlinear transition cost.\\

\noindent The two near-degenerate pairs (sc01 $\leftrightarrow$ sc02 and sc03 $\leftrightarrow$ sc04, each with Hamming distance 1) produce the smallest norms and the flattest trajectories. Because these pairs differ on only a single descriptor, the implied perturbation is concentrated on that descriptor alone. The fourteen-descriptor propagation structure contributes little, and the norm curves remain low throughout the horizon. The near-identical forward and reverse curves within each pair confirm that the network treats single-descriptor transitions almost symmetrically. Even at $H = 1$, a small directional asymmetry is detectable in the fossil-locked pair, illustrating that the framework captures directionality that the Hamming distance (as employed in \citet{Ross2026}) cannot. Full values are in Supplementary Table~\ref{tab:budgets}.\\

\begin{figure}[!ht]
    \centerline{%
    \includegraphics[keepaspectratio]{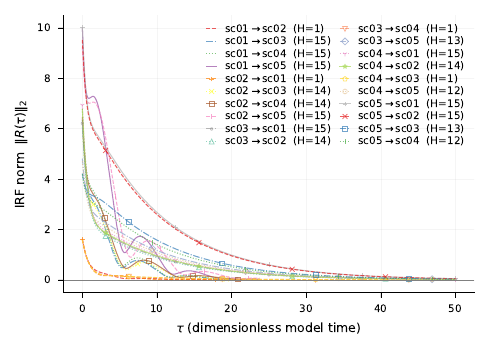}
    }
    \caption{Transition effort and cluster structure across all 20 scenario pairs. IRF norm $\|R(\tau)\|_2$ for all 20 directed scenario pairs. Each curve represents one directed transition and is labelled by origin and target scenario code together with the Hamming distance $H$. Curves partition into three natural groups: (i) near-degenerate pairs (sc01$\leftrightarrow$sc02 and sc03$\leftrightarrow$sc04, $H = 1$), with the smallest norms; (ii) within-cluster and intermediate transitions; (iii) cross-cluster transitions between the locked-in-fossil cluster (sc01, sc02) and the net-zero scenario (sc05), with the largest norms. See Supplementary Table~\ref{tab:budgets} for the complete budget values.}
    \label{fig:irf_ensemble}
\end{figure}

\noindent The cross-cluster transitions between the locked-in-fossil cluster (sc01, sc02) and the net-zero scenario (sc05) produce the highest norms, with peak values of 9--10. These curves also exhibit a faster early decay for transitions originating at sc01 or sc02, whose dominant timescales are shorter than sc05's slow ten-unit mode, whilst transitions originating at sc05 show a correspondingly slower, longer-tailed decay. This qualitative difference in curve shape by origin scenario is a direct consequence of the attractor's eigenstructure.\\

\noindent Budget asymmetry is directly readable from the norm curves in Figure~\ref{fig:irf_ensemble}. The sc05$\to$sc03 direction (budget $\|\widetilde{\boldsymbol{\varepsilon}}\| = 4.22$) meets around 14\,\% less resistance than the reverse (sc03$\to$sc05, budget $4.82$), illustrating that the direction of a transition matters and is captured by the framework. Pairs with the same maximum Hamming distance of 15 span a budget range from roughly 6 to roughly 10, confirming that the perturbation budget captures structural information entirely absent from the Hamming distance alone. Near-symmetry in both directions, as seen for the sc01$\leftrightarrow$sc05 polar pair, is an empirical property of this CIM rather than a structural requirement of the framework. Full values for all twenty pairs are given in Supplementary Table~\ref{tab:budgets}.\\

\noindent The calibration-horizon sweep in Section~\ref{sec:results:sensitivity} confirms that the three-cluster grouping is preserved across $t_1 \in [0.5, 10.0]$, with near-degenerate pairs always carrying the smallest budgets and polar cross-cluster pairs always the largest. Within clusters, reorderings occur between closely matched pairs at several horizons, reflecting the broader observation that the perturbation budget for any given transition is not a fixed quantity but varies with the calibration horizon $t_1$. Budget curves are non-monotone, falling from relatively high values at short horizons to an intermediate minimum before rising again toward the long-run limit, so each transition has a horizon at which it is least costly to pursue. Because this shape differs across pairs, the relative attractiveness of competing transition pathways shifts with the planning horizon. Standard CIB, which assigns no time dimension to scenario transitions, cannot represent any of this and provides instead a single fixed measure of inter-scenario distance insensitive to the timeframe of intervention.\\

\noindent A two-hop cross-check decomposes each polar transition into two sequential steps routed via the intermediate attractor, re-anchoring the linearisation there and summing the two local budgets. This tests whether the direct budget is sensitive to the choice of anchor point. Full details are in Section~\ref{sec:S4} of the Supplementary Materials. The two-hop estimates exceed the direct budgets, reflecting structural heterogeneity across attractors rather than numerical error, and the three-cluster grouping is preserved under re-anchoring.\\

\subsection{Indirect influence accumulation over time}
\label{sec:results:susceptibility}

\noindent Figure~\ref{fig:susceptibility} visualises the susceptibility matrix at the fossil-locked attractor (sc01) as a directed network at a short calibration horizon ($t_1 = 1$, panel a) and a long one ($t_1 = 5$, panel b). The two panels together show how indirect influence accumulates as the horizon lengthens, with node size encoding each descriptor's total outgoing influence through the network, and the contrast between panels reveals which descriptors become disproportionately leveraged over time.\\

\begin{figure}[!ht]
    \centerline{%
    \includegraphics[keepaspectratio]{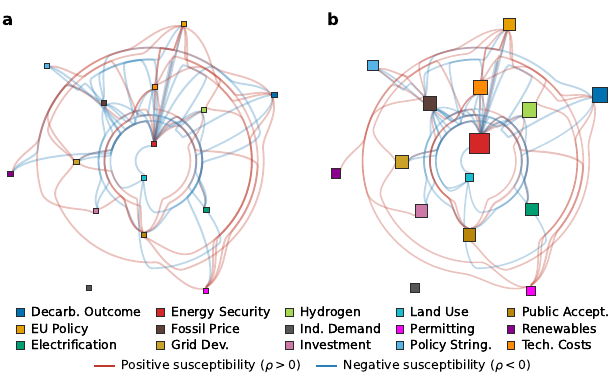}
    }
    \caption{Direct and indirect influence accumulation at sc01. Susceptibility network $\rho(t_1) = M^{-\top}(\exp(M^\top t_1) - I_N)$ at sc01 at calibration horizons $t_1 = 1$ (panel \textbf{a}, short horizon) and $t_1 = 5$ (panel \textbf{b}, long horizon). Each node is a descriptor; colours are for identification only. Node side length scales linearly with total outgoing influence (column sum of $|\rho|$: $\sum_{k \neq i}|\rho_{ki}|$), on a shared scale common to both panels, so node sizes are directly comparable across panels. The directed edge from descriptor $i$ to descriptor $k$ encodes $\rho_{ki}$: the cumulative activation of $k$ from a unit sustained push on $i$ over $[0, t_1]$ through the full network. Red edges indicate positive susceptibility; blue indicate negative. The top $25\,\%$ of all 210 directed pairs by $|\rho_{ki}|$ is shown, applied independently within each panel (about 50 edges per panel).}
    \label{fig:susceptibility}
\end{figure}

\noindent At $t_1 = 1$ (Figure~\ref{fig:susceptibility}, panel a), indirect influence has had little time to propagate and the susceptibility is concentrated in strong direct cross-impact links. The structure is relatively even. No single descriptor dominates, and most displayed edges reflect first-order interactions rather than multi-step feedback chains.\\

\noindent At $t_1 = 5$ (panel b), the network becomes substantially denser and nodes grow larger, reflecting the accumulation of indirect influence. In house-of-cards terms, the susceptibility analysis identifies which factor in the fossil-locked configuration carries the most structural weight, understood as the descriptor whose movement triggers the widest cascade of adjustments through the surrounding structure. Energy Security Pressure is that factor, with column-sum $\sum_{k \neq i}|\rho_{ki}| = 9.94$, the largest across all descriptors. A unit perturbation of this descriptor generates the largest cumulative activation in EU Policy Alignment, Public Acceptance, and Permitting Pace, and a substantial reversal in Fossil Price Pressure. Raw CIM scores alone would not identify this property, encoding only the first step of propagation. The long-horizon susceptibility network instead exposes the full indirect cascade and ranks descriptors by their network leverage.\\

\noindent The contrast between the two panels in Figure~\ref{fig:susceptibility} also illustrates why the calibration horizon $t_1$ is a substantively important modelling choice. At $t_1 = 1$, interventions appear largely independent and their collateral effects on other descriptors are small. At $t_1 = 5$, the same interventions trigger substantial collateral movement across the network. The calibration-horizon sweep in Section~\ref{sec:results:sensitivity} shows that the transition from sparse to dense indirect-influence structure occurs progressively, and that Energy Security Pressure's dominance as the leading indirect lever is present across the full range of calibration horizons examined.\\

\subsection{Sensitivity and robustness}
\label{sec:results:sensitivity}

\noindent Three modelling parameters are treated as free choices throughout: the IO-3 rescaling target $\rho_{\mathrm{target}}$, the calibration horizon $t_1$, and the Monte Carlo noise scale $\sigma$. In all three cases, qualitative findings are robust. Full details are in Section~\ref{sec:SA} of the Supplementary Materials. Taken together, the robustness results confirm the house-of-cards picture. Small imprecisions in where individual cards are placed do not change which configurations are fundamentally different from one another, and the three-cluster ranking holds throughout.\\

\section{Discussion and Conclusions}
\label{sec:discussion}

\noindent This paper extends standard Cross-Impact Balance (CIB) analysis through Linear Response Theory (LRT), deriving from the elicited cross-impact matrix alone four closed-form analytical objects that existing CIB approaches do not provide. These are the perturbation budget, the Type I cross-impact multiplier, the impulse response function, and the unit-impulse shock profile. Obtaining all four without a succession step distinguishes CIB-LRT from succession-based and Markov-chain approaches, which add pathway dynamics but treat expert-elicited qualitative scores as identified propagation coefficients, a conflation that CIB-LRT avoids. The following paragraphs discuss the contribution of each object in turn before addressing limitations and directions for further development.\\

\noindent The impulse response function (IRF) is a key contribution to CIB practice. It recovers the ordering in which descriptors respond to a transition pressure from the eigenstructure of the drift matrix, an intrinsic network invariant that succession-based approaches cannot produce because their synchronous update protocol imposes the same timescale on every descriptor. Succession-based update steps also carry no elicited calendar meaning, since the duration of each iteration is not defined within the CIB framework, a limitation common to all dynamic CIB approaches, addressed as the fourth limitation below. The IRF can also expose sign reversals, in which a descriptor's activation pressure temporarily moves counter to its ultimate transition direction before correcting, a consequence of feedback loops in the network that are invisible in the raw cross-impact scores. Identifying such reversals in advance matters directly for monitoring framework design.\\

\noindent The Type I cross-impact multiplier identifies reversals of this kind. A negative diagonal entry $\Lambda_{jj}$ (the cumulative long-run activation that descriptor $j$ accumulates from sustained pressure on itself, propagated through all indirect feedback chains in the network) indicates, within the Ornstein-Uhlenbeck (OU) model, that this pressure ultimately reverses the descriptor's own long-run activation. Because this property differs across consistent scenarios, the multiplier provides an attractor-specific structural fingerprint that no other CIB object supplies, allowing practitioners to anticipate which descriptor trajectories will overshoot their target direction before eventually correcting. The unit-impulse shock profile extends this attractor-specificity further, characterising how the same external push propagates differently depending on which structural equilibrium the system occupies.\\

\noindent The perturbation budget complements trajectory objects by quantifying transition resistance as a directed, network-weighted scalar, capturing structural information missed by the simple Hamming distance (the count of differing descriptors between two scenarios). Unlike the symmetric Hamming distance, the perturbation budget is asymmetric across directed scenario pairs, reflecting directional preferences inherent in the cross-impact structure at the origin attractor. The required effort for a transition, however, is not a fixed value but depends on the timeframe in which it is undertaken, and varies between scenario pairs. Thus, practitioners gain insights beyond what the Hamming distance offers, as standard analysis records only the number of differences, not how transition effort evolves with the planning horizon. These capabilities are illustrated in the energy-transition application.\\

\noindent The CIB-LRT framework and dynamic pathway simulation address different questions about the same cross-impact matrix and are designed to complement rather than compete with one another. The former produces closed-form structural invariants, while the latter reveals which attractors the system reaches under shocks, alternative regime matrices, and uncertain scores. Used in combination, the two approaches are mutually reinforcing. The multiplier identifies the most effective indirect levers at each attractor, guiding the selection of descriptors to target when constructing regime matrices or applying structural shocks. The perturbation budget ranks all directed attractor pairs by network-weighted transition effort, providing a principled basis for choosing which transitions to stress-test. An unexpected attractor endpoint surfaced by pathway simulation can in turn be characterised analytically. The budget quantifies how much sustained network forcing that transition requires, and the impulse response shows which descriptors bear the pressure first.\\

\noindent Four limitations should be borne in mind. The Ornstein-Uhlenbeck embedding is adopted as the minimal continuous-time embedding consistent with the CIB drift structure rather than derived from CIB first principles (Section~\ref{sec:methods:ou}). As shown there, the stationary covariance cancels from the susceptibility derivation, so the four analytical objects depend only on the drift matrix $M$ and are therefore not sensitive to the specific noise structure of the embedding. They are, however, structural properties of the linearisation at the reference attractor rather than empirically identified predictions of adjustment speed or transition cost. All four LRT objects are local approximations formally valid near the reference equilibrium. Nevertheless, the two-hop cross-check in the Supplementary Materials confirms that broad rankings are preserved under re-anchoring, but results involving large scenario displacements should be read as structural indicators rather than precise trajectory predictions. The three free modelling parameters (the IO-3 rescaling target spectral radius $\rho_{\mathrm{target}}$, the calibration horizon $t_1$, and the Monte Carlo noise amplitude $\sigma$) should always be reported and sensitivity-tested alongside results, as Section~\ref{sec:results:sensitivity} demonstrates. A fourth limitation extends beyond CIB-LRT to all dynamic CIB approaches that do not explicitly elicit a time unit from the expert panel. The time axis is dimensionless, and converting it to calendar duration requires empirical grounding that the elicitation cannot supply. Succession-based and Markov-chain methods face the same gap, which the existing literature has not acknowledged, and their synchronous protocol additionally imposes a uniform iteration timescale that masks the real-world heterogeneity of descriptor dynamics, since descriptors (e.g. regulatory frameworks, price signals, and infrastructure investments) evolve at fundamentally different speeds. The partial remedy of a single expert-elicited calibration anchor is discussed in Section~\ref{sec:methods:ou}.\\

\noindent Beyond energy-system transitions, the framework applies to any domain where a standard CIB elicitation has been conducted, including urban planning, public-health system transitions, corporate strategy, and climate-adaptation governance. Several natural extensions are worth noting. A single temporal calibration anchor, elicited from the expert panel, would fix the location of the first IRF peak in calendar time without altering the eigenstructure-derived ordering of descriptor responses. The framework also extends naturally to the probabilistic CIB setting of \citet{Salo2022} and \citet{Roponen2024}, where CIM entries carry full uncertainty distributions, allowing the LRT objects to be computed across the posterior. Both the LRT objects and the companion pathway framework of \citet{Ross2026} are implemented in PyCIB \citep{RossAG2026pycib}, making the full analytical toolkit accessible to any team engaged in CIB elicitation.\\

\section*{Data statement} \noindent Data and model code are available at: \url{https://github.com/ag-ross/CIBLRT}.

\section*{Funding} \noindent This work was supported in part by the Helmholtz Association under the programme ``Energy System Design''. Open Access is funded by the Deutsche Forschungsgemeinschaft (DFG, German Research Foundation) (491111487).


\clearpage

\setcounter{section}{0}
\setcounter{subsection}{0}
\setcounter{figure}{0}
\setcounter{table}{0}
\setcounter{equation}{0}
\renewcommand{\thesection}{S\arabic{section}}
\renewcommand{\thesubsection}{S\arabic{section}.\arabic{subsection}}
\renewcommand{\thefigure}{S\arabic{figure}}
\renewcommand{\thetable}{S\arabic{table}}
\renewcommand{\theequation}{S\arabic{equation}}

\begin{center}
    {\normalsize\textbf{Supplementary Material}\par}
    \vspace{0.25cm}
    {\normalsize \manuscripttitle\par}
\end{center}
\vspace{0.5cm}

\noindent This supplementary document is organised as follows. Section~\ref{sec:SA} reports three systematic sensitivity analyses corresponding to the three free modelling parameters of the main analysis, namely the IO-3 rescaling target $\rho_{\mathrm{target}}$ (Section~\ref{sec:S1}), the calibration horizon $t_1$ (Section~\ref{sec:S2}), and the Monte Carlo noise scale $\sigma$ (Section~\ref{sec:S3}). Together, they support the robustness claims summarised in Section~\ref{sec:results:sensitivity} of the main paper. Section~\ref{sec:S4} reports a two-hop path-decomposition cross-check of anchor sensitivity, which is a validity test of a different kind rather than a sensitivity analysis of a free parameter. Section~\ref{sec:S5} gives the full cross-impact multiplier diagonal entries across all descriptors and scenarios. Section~\ref{sec:S6} provides a self-contained derivation of the susceptibility formula, including a proof that the stationary covariance $\Sigma$ cancels exactly from the Green-Kubo integral, leaving the four analytical objects depending only on the drift matrix $M$. Supplementary Table~\ref{tab:budgets} gives the complete perturbation budgets for all 20 directed scenario transitions. Supplementary Table~\ref{tab:lambda_diag} gives the diagonal entries $\Lambda_{jj}$ of the cross-impact multiplier for all fifteen descriptors at each attractor.

\section{Sensitivity Analysis}
\label{sec:SA}

\noindent The three subsections below each vary one free modelling parameter and assess the robustness of key outputs to that choice.\\

\subsection{IO-3 Rescaling Target}
\label{sec:S1}

\noindent Supplementary Figure~\ref{fig:S1} shows four key LRT outputs as functions of $\rho_{\mathrm{target}} \in [0.50, 0.99]$ evaluated on a 50-point grid. Two features are consistent across all panels. First, the qualitative ordering of scenarios (which scenario has the largest multiplier norm, which transition has the highest budget) is stable across the full range, so the rankings reported in the main analysis do not depend on the choice of $\rho_{\mathrm{target}}$. Second, the absolute magnitudes of all four outputs increase sharply as $\rho_{\mathrm{target}} \to 1$, since the system approaches marginal stability and both the multiplier and the susceptibility diverge whilst the perturbation budget grows without bound. This confirms that outputs should be read as relative rather than absolute measures, and that the analyst's choice of $\rho_{\mathrm{target}}$ must be treated explicitly as a modelling decision and reported alongside the results.\\

\begin{figure}[!ht]
    \centerline{%
    \includegraphics[keepaspectratio]{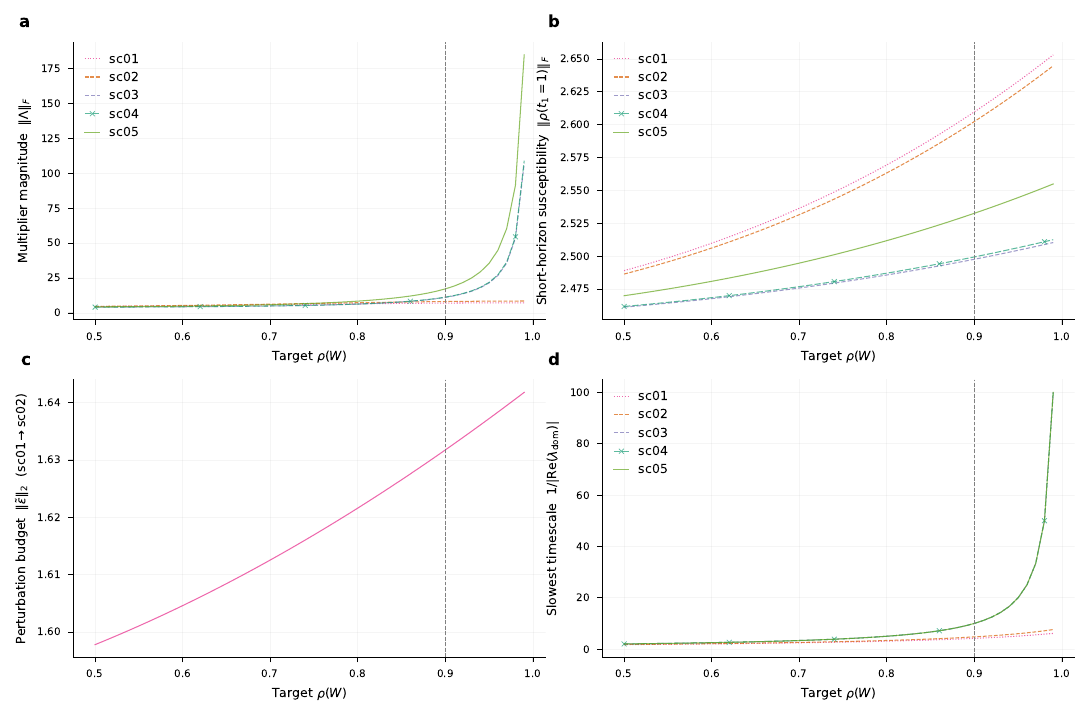}
    }
    \caption{LRT output sensitivity to the IO-3 rescaling target. Sensitivity of key outputs to the IO-3 rescaling target $\rho_{\mathrm{target}}$, evaluated on a 50-point grid from 0.50 to 0.99. Panel \textbf{a} (top left): Frobenius norm $\|\Lambda\|_F$ of the cross-impact multiplier. Panel \textbf{b} (top right): short-horizon susceptibility norm $\|\rho(t_1{=}1)\|_F$. Panel \textbf{c} (bottom left): perturbation budget $\|\widetilde{\boldsymbol{\varepsilon}}\|$ for the sc01$\to$sc02 transition. Panel \textbf{d} (bottom right): slowest relaxation timescale $1/|\operatorname{Re}(\lambda_{\mathrm{dom}})|$, where $\lambda_{\mathrm{dom}}$ is the eigenvalue of $M$ with least-negative real part (slowest-decaying mode). The vertical dashed line marks the nominal value $\rho_{\mathrm{target}} = 0.9$ used in the main analysis.}
    \label{fig:S1}
\end{figure}

\subsection{Calibration Horizon Sweep}
\label{sec:S2}

\noindent Supplementary Figure~\ref{fig:S2}, panel~a shows the perturbation budget $\|\widetilde{\boldsymbol{\varepsilon}}(t_1)\|$ for all 20 directed scenario pairs as a function of the calibration horizon $t_1 \in [0.5, 10.0]$ (20 equally spaced evaluation points). Budget curves are not monotone over the plotted range. They decrease steeply from the short-horizon limit, reach a minimum (typically around $t_1 \approx 3$--$5$ for most pairs in this dataset), and then increase gradually back toward a finite asymptote as the susceptibility converges to the long-run multiplier. The nominal value $t_1 = 1.0$ used in the main analysis is chosen within the short-to-medium horizon range, where network propagation is active and qualitative differences between scenario pairs are clearly readable across the full range of scenario pairs. By comparison, the minimum-budget horizon ($t_1 \approx 3$--$5$) produces identical cluster rankings but yields lower absolute values. This behaviour is captured by the sweep in panel~(a).\\

\noindent The principal finding is that the three-cluster grouping of scenario pairs is stable across the full range of $t_1$. Near-degenerate pairs (sc01$\leftrightarrow$sc02 and sc03$\leftrightarrow$sc04, Hamming distance $H = 1$) always produce the smallest budgets and the sc01$\leftrightarrow$sc05 pair always produces the largest. Nevertheless, within clusters, reorderings occur at several horizons between pairs with closely similar budgets. This establishes the broad budget ranking as an intrinsic property of the cross-impact network, independent of the calibration horizon choice, and confirms that the key qualitative finding (network structure, not Hamming distance alone, determines transition cost) holds across all modelling assumptions examined.\\

\begin{figure}[!ht]
    \centerline{%
    \includegraphics[keepaspectratio]{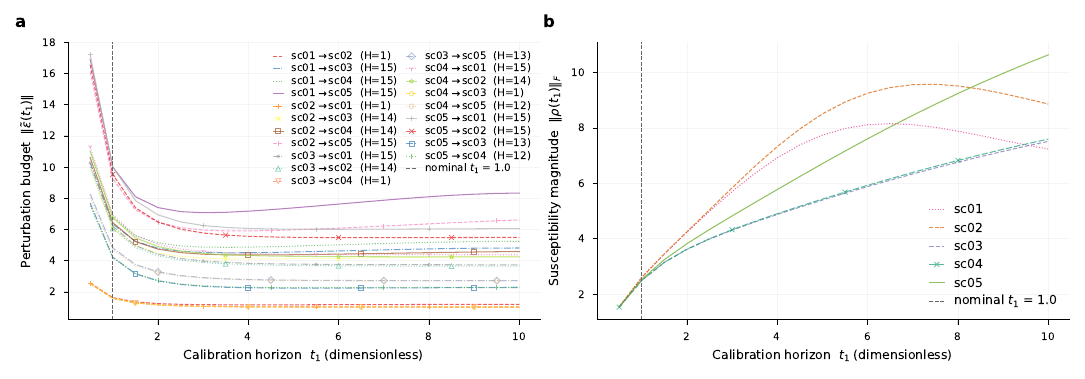}
    }
    \caption{Budget and susceptibility norm sensitivity to calibration horizon. Calibration-horizon sweep over $t_1 \in [0.5, 10.0]$ (20 evaluation points; $t_1$ in dimensionless model units). Panel \textbf{a} (left): perturbation budget $\|\widetilde{\boldsymbol{\varepsilon}}(t_1)\|$ for all 20 directed scenario pairs. Each curve corresponds to one directed transition, labelled by origin and target scenario with Hamming distance $H$. Panel \textbf{b} (right): Frobenius norm $\|\rho(t_1)\|_F$ of the susceptibility matrix, shown separately for each of the five consistent scenarios. In both panels the vertical dashed line marks the nominal value $t_1 = 1.0$ used in the main analysis.}
    \label{fig:S2}
\end{figure}

\subsection{Monte Carlo Noise-Scale Sensitivity}
\label{sec:S3}

\noindent Supplementary Figure~\ref{fig:S3} examines how the perturbation budget varies with the Monte Carlo noise scale $\sigma$ for the representative transition sc05$\to$sc03 (analysed in Section~\ref{sec:results:irf} of the main paper). At each of ten noise levels $\sigma \in [0.1, 1.0]$, $10{,}000$ independent realisations were generated, matching the sample size used in the Figure~\ref{fig:irf_sc05_sc03} ribbon and the Figure~\ref{fig:shock_sc05_policy} unit-impulse analysis (in the main paper). Each $\sigma$ value uses its own RNG seed, so the realisations are independent across noise levels rather than the same draws rescaled. Each realisation adds Gaussian noise $\mathcal{N}(0, \sigma^2)$ independently to every off-diagonal CIM entry, re-applies IO-3 rescaling, and re-computes the perturbation budget via the implied perturbation formula of the main paper. The standard error on the mean budget at $N = 10{,}000$ is around $0.0015$, and the standard error on the IQR/mean estimate is around $0.05$ percentage points, both small enough that the trend across $\sigma$ values is resolved well below plot resolution. All $10{,}000$ realisations were accepted at every noise level (zero rejection rate). The zero rejection rate confirms that IO-3 rescaling renders the stability condition non-binding for this cross-impact matrix (the condition is satisfied with a large margin throughout the uncertainty range tested).\\

\noindent The mean budget remains at around 4.22 across the full range $\sigma \in [0.1, 1.0]$, shifting by less than $0.1\,\%$ from the deterministic value at any noise level tested. By contrast, the interquartile range grows linearly with $\sigma$, reaching around $4.1\,\%$ of the mean at $\sigma = 1.0$, consistent with a first-order perturbation response and with no evidence of nonlinear threshold behaviour.\\

\noindent These results establish the perturbation budget as a robust property of the cross-impact matrix. The budget is not sensitive to the specific numerical scores agreed by the experts. Rather, it reflects the network's topological structure (which descriptors are connected and with what sign) rather than the precise magnitude of the agreed scores. This is an operationally important finding: practitioners can use the budget with confidence even when individual scores carry substantial uncertainty, as is common in expert-elicitation exercises.\\

\begin{figure}[!ht]
    \centerline{%
    \includegraphics[keepaspectratio]{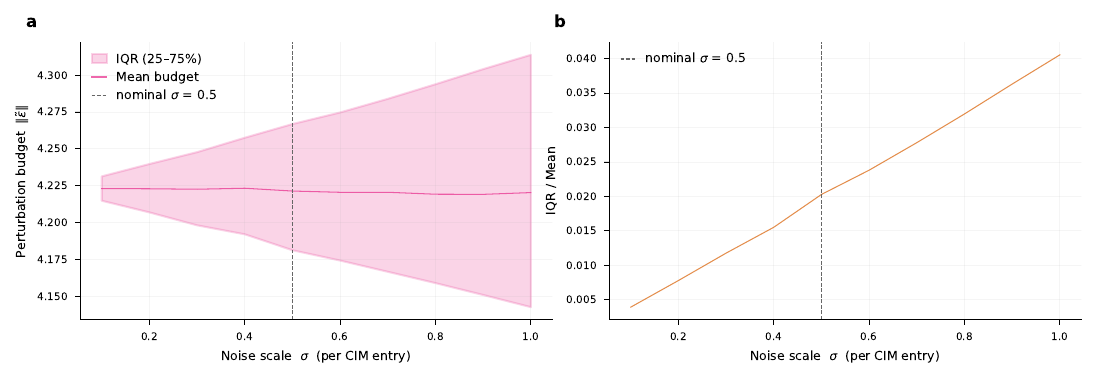}
    }
    \caption{Perturbation budget robustness to score uncertainty. Sensitivity of the perturbation budget to Monte Carlo noise scale $\sigma$ for the transition sc05$\to$sc03 ($N = 10{,}000$ realisations per noise level; all realisations accepted). Panel \textbf{a} (left): budget distribution vs.\ noise scale, showing mean budget (line) and 25th--75th percentile ribbon (shaded) as a function of $\sigma$. Panel \textbf{b} (right): relative budget uncertainty vs.\ noise scale, showing IQR / mean as a function of $\sigma$. The vertical dashed line marks the nominal value $\sigma = 0.5$ used in the main analysis.}
    \label{fig:S3}
\end{figure}

\section{Two-Hop Path-Decomposition Cross-Check}
\label{sec:S4}

\noindent The LRT perturbation budget is the exact answer to a local structural question: how much network-weighted forcing does the linearised cross-impact structure at the origin attractor require to achieve the prescribed displacement? For small perturbations this local answer also approximates the full nonlinear dynamics well, but for the large-displacement cross-cluster transitions examined here (Hamming distances of 12--15 out of 15 descriptors), the linear model is extrapolating far beyond its formal regime of validity. The budget nonetheless remains a well-defined structural property of the origin attractor's network topology. There is no nonlinear ground truth against which it could be declared ``wrong''. The budget characterises the origin attractor's structure, and a different attractor's structure would give a different (equally valid) answer.\\

\noindent The two-hop cross-check reported here tests sensitivity to the choice of linearisation anchor. For any cross-cluster transition $A \to B$, the two-hop decomposed budget via an intermediate attractor $C$ is $\mathcal{B}(A{\to}C) + \mathcal{B}(C{\to}B)$ (where $\mathcal{B}(X{\to}Y) \equiv \|\widetilde{\boldsymbol{\varepsilon}}_{X\to Y}\|$ denotes the perturbation budget for the transition $X{\to}Y$), with each hop using the local linearisation at its own origin attractor. Both the direct and the two-hop estimates are exact answers to their respective local structural questions, differing because the cross-impact structure genuinely changes across attractor contexts. The comparison therefore reveals which budget features are robust to re-anchoring and which are not, rather than quantifying error relative to a ground truth. The budget values used below are taken from Table~\ref{tab:budgets}.\\

\noindent For all four polar fossil--net-zero pairs (sc01/sc02\,$\leftrightarrow$\,sc05), the direct budget is consistently below the two-hop re-anchored estimate, with gaps ranging from 4 to 21\,\%. Re-anchoring at an intermediate attractor raises the estimated cost in every case, reflecting the structural heterogeneity between the polar and intermediate attractors.\\

\noindent Two additional findings from this comparison warrant explicit note. First, the sc01$\leftrightarrow$sc05 pair is symmetric in direct budgets ($10.01$ each way) but asymmetric under two-hop substitution (sc05$\to$sc01 via sc03: $10.43$; sc01$\to$sc05 via sc03: $11.18$), revealing that the apparent symmetry is a feature of the local linearisation at each polar attractor rather than an intrinsic structural property. It breaks when intermediate attractors serve as anchor points, demonstrating sensitivity to the choice of linearisation anchor. Second, the sc02$\leftrightarrow$sc05 asymmetry reverses. The direct budgets suggest moving from sc02 to sc05 is cheaper ($9.25$ vs $9.53$), but the two-hop estimates suggest the opposite (sc05$\to$sc02: $10.28$; sc02$\to$sc05: $10.88$). This reversal indicates that the small sc02$\leftrightarrow$sc05 asymmetry in the direct budgets is not robust to the choice of linearisation anchor, and small asymmetries between nearly symmetric polar pairs should not be over-interpreted.\\

\noindent The three-cluster grouping is robust throughout. Near-degenerate pairs remain the cheapest (direct: 1.59--1.63) and polar fossil--net-zero pairs the most expensive even under two-hop substitution (minimum polar two-hop: $10.28$, well above the maximum middle cross-cluster direct budget of $6.94$). The broad budget ranking is therefore a structural property of the cross-impact network. For practical use, direct budgets for the polar fossil--net-zero pairs (sc01/sc02\,$\leftrightarrow$\,sc05) are consistently below the corresponding two-hop re-anchored estimates by 4--21\,\%, reflecting structural heterogeneity between origin and intermediate attractors, and small asymmetries within the polar group should be treated with caution. For other large-displacement pairs (fossil--intermediate or intermediate--net-zero transitions), no natural intermediate attractor exists between the endpoints, so the two-hop decomposition does not provide a comparable locality bound and the locality error for those transitions is unconstrained by this analysis.\\

\section{Type I Cross-Impact Multiplier Diagonal Entries}
\label{sec:S5}

\noindent Table~\ref{tab:lambda_diag} reports the diagonal entries $\Lambda_{jj}$ of the cross-impact multiplier $(I_N - W)^{-\top}$ for all fifteen descriptors at each of the five consistent scenarios ($\rho_{\mathrm{target}} = 0.9$). Entry signs are invariant to $\rho_{\mathrm{target}}$ (confirmed numerically across the full range $[0.50, 0.99]$, as shown in Section~\ref{sec:S1}), whilst magnitudes scale with it. Negative entries indicate descriptors for which sustained pressure ultimately reverses its own long-run activation through indirect feedback in the cross-impact network.\\

\section{Derivation of the CIB-LRT Analytical Objects}
\label{sec:S6}

\noindent A key technical observation enabling the CIB-LRT framework is that the stationary covariance $\Sigma$ of the Ornstein-Uhlenbeck (OU) embedding cancels exactly from the Green-Kubo integral, so all four analytical objects depend only on the drift matrix $M = W - I_N$ and require no estimation of noise parameters. This section gives self-contained proofs of the susceptibility matrix (Proposition~\ref{prop:susceptibility}), the impulse response function (Corollary~\ref{cor:irf}), and the Type I cross-impact multiplier (Corollary~\ref{cor:multiplier}); the perturbation budget follows directly from the susceptibility by inversion and requires no separate derivation. The existence of the stationary distribution and its Lyapunov characterisation (equation~\eqref{eq:S_lyapunov} below) are standard results that are cited rather than proved here. All four objects reduce to matrix exponential evaluations requiring no simulation.\\

\subsection*{Assumptions and stationary distribution}

\noindent Let $M \in \mathbb{R}^{N \times N}$ be the drift matrix with all eigenvalues having strictly negative real parts (the stability condition; see Section~\ref{sec:methods:ou} of the main paper). The OU process

\begin{equation}
  \mathrm{d}\mathbf{x}(t) = M\mathbf{x}(t)\,\mathrm{d}t + \mathrm{d}\mathbf{B}(t)
  \label{eq:S_SDE}
\end{equation}

then admits a unique stationary distribution, which is multivariate Gaussian with mean zero and covariance matrix $\Sigma \in \mathbb{R}^{N \times N}$ satisfying the Lyapunov equation

\begin{equation}
  M\Sigma + \Sigma M^\top + I_N = 0.
  \label{eq:S_lyapunov}
\end{equation}

(\citealt{Risken1989}.) The stability condition ensures $\Sigma$ is symmetric positive definite and hence invertible.\\

\subsection*{Equilibrium cross-covariance}

\begin{proposition}
  \label{prop:crosscov}
  Under the stability condition, the equilibrium cross-covariance function of the stationary process satisfies, for all $\tau \geq 0$,
  \begin{equation}
    \langle\mathbf{x}(0)\,\mathbf{x}(\tau)^\top\rangle_0 = \Sigma\exp(M^\top\tau).
    \label{eq:S_crosscov}
  \end{equation}
\end{proposition}

\begin{proof}
  The formal solution of~\eqref{eq:S_SDE} is $\mathbf{x}(\tau) = \exp(M\tau)\mathbf{x}(0) + \int_0^\tau \exp(M(\tau-s))\,\mathrm{d}\mathbf{B}(s)$.
  In the stationary regime the stochastic integral is independent of $\mathbf{x}(0)$ (future increments of $\mathbf{B}$ are independent of the current state), so
  \begin{align*}
    \langle\mathbf{x}(0)\,\mathbf{x}(\tau)^\top\rangle_0
      &= \langle\mathbf{x}(0)\bigl(\exp(M\tau)\mathbf{x}(0)\bigr)^\top\rangle_0 \\
      &= \langle\mathbf{x}(0)\,\mathbf{x}(0)^\top\rangle_0\,\exp(M^\top\tau)
       = \Sigma\exp(M^\top\tau). \qedhere
  \end{align*}
\end{proof}

\noindent The argument uses the standard variation-of-constants representation of the OU process; cf.\ \citet{Raseta} for a related treatment in the IO economics setting. Note that \citet{Raseta} states the equilibrium cross-covariance (using their notation $Q(\tau)$ for this quantity) in the form $\exp(M\tau)\Sigma$, whereas the present derivation gives $\Sigma\exp(M^\top\tau)$; the two expressions are transposes of one another. The labels for $Q(\tau)$ and its transpose are inverted in \citet{Raseta}, though the final susceptibility formula there survives because the susceptibility proof uses the correctly-valued quantity.\\

\subsection*{Susceptibility formula and the \texorpdfstring{$\Sigma$}{Sigma}-cancellation}

\noindent The Green-Kubo susceptibility matrix at horizon $t_1 > 0$ is defined as
\begin{equation}
  \rho(t_1) = \Sigma^{-1}\int_0^{t_1}\langle\mathbf{x}(0)\,\mathbf{x}(\tau)^\top\rangle_0\,\mathrm{d}\tau.
  \label{eq:S_GK}
\end{equation}

\begin{proposition}
  \label{prop:susceptibility}
  Under the stability condition,
  \begin{equation}
    \rho(t_1) = M^{-\top}\bigl(\exp(M^\top t_1) - I_N\bigr).
    \label{eq:S_susc}
  \end{equation}
\end{proposition}

\begin{proof}
  Substituting Proposition~\ref{prop:crosscov} into~\eqref{eq:S_GK}:
  \begin{align}
    \rho(t_1)
      &= \Sigma^{-1}\int_0^{t_1}\Sigma\exp(M^\top\tau)\,\mathrm{d}\tau \notag \\
      &= \int_0^{t_1}\exp(M^\top\tau)\,\mathrm{d}\tau \label{eq:S_cancel}\\
      &= M^{-\top}\bigl(\exp(M^\top t_1) - I_N\bigr). \notag
  \end{align}
  The $\Sigma^{-1}$ and $\Sigma$ cancel in step~\eqref{eq:S_cancel}. The stability condition ensures $M^\top$ is invertible, so the matrix exponential integral evaluates as stated \citep{higham2008functions}.
\end{proof}

\noindent The cancellation at step~\eqref{eq:S_cancel} is the key technical fact of the CIB-LRT application: the specific noise structure of the OU embedding enters the Green-Kubo integral only through $\Sigma$, and $\Sigma$ disappears. The result~\eqref{eq:S_susc} is determined entirely by the drift matrix $M$, which is in turn determined by the CIM at the reference attractor. A related result in the IO economics setting appears in \citet{Raseta}; the proofs presented here are self-contained and do not depend on that work.

\subsection*{Impulse response function}

\begin{corollary}
  \label{cor:irf}
  Under the stability condition, the impulse response function at lag $\tau \geq 0$ is
  \begin{equation}
    R(\tau) = \exp(M^\top\tau)\,\widetilde{\boldsymbol{\varepsilon}},
    \label{eq:S_IRF}
  \end{equation}
  where $\widetilde{\boldsymbol{\varepsilon}} = \rho(t_1)^{-1}\Delta\mathbf{z}$ is the implied perturbation vector.
\end{corollary}

\begin{proof}
  In the Green-Kubo framework \citep{Green1954,Kubo1957,Klimek2019}, the susceptibility is the integral $\rho(t_1) = \int_0^{t_1}\mathcal{R}(\tau)\,\mathrm{d}\tau$ of the time-resolved response kernel $\mathcal{R}(\tau) = \Sigma^{-1}\langle\mathbf{x}(0)\,\mathbf{x}(\tau)^\top\rangle_0$, as in equation~\eqref{eq:S_GK}. The IRF is this kernel applied to $\widetilde{\boldsymbol{\varepsilon}}$:
  \[R(\tau) = \mathcal{R}(\tau)\,\widetilde{\boldsymbol{\varepsilon}} = \Sigma^{-1}\langle\mathbf{x}(0)\,\mathbf{x}(\tau)^\top\rangle_0\,\widetilde{\boldsymbol{\varepsilon}}.\]
  Substituting Proposition~\ref{prop:crosscov}:
  \[R(\tau) = \Sigma^{-1}\cdot\Sigma\exp(M^\top\tau)\cdot\widetilde{\boldsymbol{\varepsilon}} = \exp(M^\top\tau)\,\widetilde{\boldsymbol{\varepsilon}}.\]
  The $\Sigma^{-1}$ and $\Sigma$ cancel by the same mechanism as Proposition~\ref{prop:susceptibility}.
\end{proof}

\noindent It follows immediately that $\int_0^{t_1}R(\tau)\,\mathrm{d}\tau = \rho(t_1)\widetilde{\boldsymbol{\varepsilon}} = \Delta\mathbf{z}$, confirming that the IRF integrates exactly to the target displacement over the calibration horizon.

\subsection*{Type I cross-impact multiplier}

\begin{corollary}
  \label{cor:multiplier}
  Under the stability condition,
  \begin{equation}
    \Lambda := \lim_{t_1\to\infty}\rho(t_1) = -M^{-\top} = (I_N - W)^{-\top}.
    \label{eq:S_multiplier}
  \end{equation}
\end{corollary}

\begin{proof}
  Since all eigenvalues of $M$ have strictly negative real parts, $\exp(M^\top t_1)\to 0$ as $t_1\to\infty$. Taking the limit in Proposition~\ref{prop:susceptibility}:
  \[\Lambda = M^{-\top}(0 - I_N) = -M^{-\top}.\]
  Substituting $M = W - I_N$ gives $-M^{-\top} = -(W - I_N)^{-\top} = (I_N - W)^{-\top}.$
\end{proof}

\noindent Corollaries~\ref{cor:irf} and~\ref{cor:multiplier}, together with Propositions~\ref{prop:crosscov} and~\ref{prop:susceptibility}, show that all four analytical objects of the CIB-LRT framework are determined entirely by the drift matrix $M$ and reduce to evaluations of the matrix exponential $\exp(M^\top\tau)$ or its integral. No simulation of the OU process is required: the exact closed-form expressions can be evaluated in $O(N^3)$ time using standard algorithms \citep{higham2008functions}, regardless of required precision or matrix dimension.\\


\newpage

\begin{table}[htbp]
\centering
\caption{Perturbation budgets $\|\widetilde{\boldsymbol{\varepsilon}}\|$ for all 20 directed scenario transitions ($t_1 = 1.0$, $\rho_{\mathrm{target}} = 0.9$).}
\label{tab:budgets}
\footnotesize
\begin{tabular}{llccr}
\toprule
From & To & $H$ & $\|\Delta\mathbf{z}\|$ & Budget $\|\widetilde{\boldsymbol{\varepsilon}}\|$ \\
\midrule
sc01 & sc02 &  1 & 1.00 & 1.63 \\
sc01 & sc03 & 15 & 4.24 & $6.36^\dagger$ \\
sc01 & sc04 & 15 & 4.58 & $6.79^\dagger$ \\
sc01 & sc05 & 15 & 7.55 & $10.01^\dagger$ \\
sc02 & sc01 &  1 & 1.00 & 1.59 \\
sc02 & sc03 & 14 & 4.12 & $6.06^\dagger$ \\
sc02 & sc04 & 14 & 4.47 & $6.44^\dagger$ \\
sc02 & sc05 & 15 & 7.35 & $9.25^\dagger$ \\
sc03 & sc01 & 15 & 4.24 & $6.21^\dagger$ \\
sc03 & sc02 & 14 & 4.12 & $6.06^\dagger$ \\
sc03 & sc04 &  1 & 1.00 & 1.59 \\
sc03 & sc05 & 13 & 3.61 & $4.82^\dagger$ \\
sc04 & sc01 & 15 & 4.58 & $6.94^\dagger$ \\
sc04 & sc02 & 14 & 4.47 & $6.76^\dagger$ \\
sc04 & sc03 &  1 & 1.00 & 1.59 \\
sc04 & sc05 & 12 & 3.46 & $4.71^\dagger$ \\
sc05 & sc01 & 15 & 7.55 & $10.01^\dagger$ \\
sc05 & sc02 & 15 & 7.35 & $9.53^\dagger$ \\
sc05 & sc03 & 13 & 3.61 & $4.22^\dagger$ \\
sc05 & sc04 & 12 & 3.46 & $4.19^\dagger$ \\
\bottomrule
\end{tabular}
\smallskip\par\noindent
{\footnotesize Budgets are computed at the origin attractor and are asymmetric by construction. Scenario characterisations are given in Table~\ref{tab:scenarios} of the main paper. $H$: Hamming distance (number of differing descriptors). $\|\Delta\mathbf{z}\|$: Euclidean norm of the state-index displacement vector; differs from $H$ because descriptors can change by 0, 1, or 2 ordinal levels. $^\dagger$~Linearisation extrapolation beyond the formal regime of validity (Section~\ref{sec:discussion} of the main paper); two-hop cross-check (Section~\ref{sec:S4}) shows direct budgets lie 4--21\,\% below re-anchored estimates for polar fossil--net-zero pairs, reflecting structural heterogeneity across attractors rather than numerical error; three-cluster grouping preserved.}
\end{table}

\newpage

\begin{table}[htbp]
\centering
\caption{Diagonal entries $\Lambda_{jj}$ of the cross-impact multiplier $(I_N - W)^{-\top}$ at each attractor ($\rho_{\mathrm{target}} = 0.9$).}
\label{tab:lambda_diag}
\footnotesize
\begin{tabular}{lrrrrr}
\toprule
Descriptor & sc01 & sc02 & sc03 & sc04 & sc05 \\
\midrule
Decarbonisation Outcome   &  0.832 &  0.863 &  1.853 &  1.816 &  1.812 \\
EU Policy Alignment       &  1.366 &  1.395 &  1.676 &  1.736 &  2.147 \\
Electrification Pace      &  1.038 &  1.105 &  1.868 &  1.863 &  1.826 \\
Energy Security Pressure  & $-$0.430$^\dagger$ & $-$0.877$^\dagger$ &  1.233 &  1.251 & $-$0.387$^\dagger$ \\
Fossil Price Pressure     &  0.161 & $-$0.085$^\dagger$ &  1.228 &  1.244 &  0.003 \\
Grid Development          &  1.140 &  1.257 &  1.735 &  1.726 &  2.177 \\
Hydrogen Role             &  0.864 &  0.928 &  1.773 &  1.896 &  1.693 \\
Industrial Energy Demand  &  0.872 &  0.839 &  1.476 &  1.516 &  1.141 \\
Investment Availability   &  1.003 &  1.111 &  1.866 &  1.878 &  2.128 \\
Land Use Conflict         &  0.947 &  1.084 &  1.074 &  1.076 &  0.530 \\
Permitting Pace           &  1.047 &  1.078 &  1.452 &  1.160 &  2.070 \\
Policy Stringency         &  1.004 &  1.044 &  1.691 &  1.720 &  1.861 \\
Public Acceptance         &  1.279 &  1.264 &  1.511 &  1.551 &  1.897 \\
Renewables Deployment     &  0.961 &  1.012 &  1.835 &  1.866 &  2.513 \\
Technology Costs          &  1.325 &  1.498 &  0.889 &  0.893 &  2.013 \\
\bottomrule
\end{tabular}
\smallskip\par\noindent
{\footnotesize $^\dagger$ Negative entry: sustained pressure on descriptor $j$ ultimately reverses its own long-run activation through indirect feedback.}
\end{table}

\newpage

\renewcommand{\bibsection}{%
    \section*{References}
    \addcontentsline{toc}{section}{References}
}
\bibliographystyle{elsarticle-harv}
\bibliography{references}
\end{document}